\documentclass[11pt,letterpaper]{article}
\usepackage{amsmath,amssymb,amsthm,epsfig}
\usepackage{graphpap}
\usepackage{color}
\usepackage{graphicx}
\usepackage{epstopdf}
\usepackage{pdflscape}
\usepackage{subfig}
\usepackage{epsfig}
\usepackage{amsmath}
\usepackage{amssymb}
\usepackage{amsthm}
\usepackage{latexsym}
\usepackage{amsfonts}
\usepackage{bm}
\usepackage{multirow}
\usepackage{setspace}
\usepackage{times}
\usepackage[square,sort,comma,numbers]{natbib}
\usepackage[ruled,vlined,linesnumbered,commentsnumbered]{algorithm2e}

\newtheorem{theorem}{Theorem}
\newtheorem{lemma}{Lemma}
\newtheorem{observation}{Observation}
\newtheorem{definition}{Definition}

\usepackage{fullpage}

\newcommand{\ugraph}{\mathcal{G}}

\newcommand{\cliqueprob}{clq}
\newcommand{\remove}[1]{}

\begin{document}
\date{}

\title{\bf Mining Maximal Cliques from an Uncertain Graph
}

\author{
Arko Provo Mukherjee\thanks{
Dept. of ECE, Iowa State University.
Email: arko@iastate.edu
}
\and 
Pan Xu\thanks{
Dept. of Computer Science, University of Maryland.
Email: panxu@cs.umd.edu
}
\and 
Srikanta Tirthapura\thanks{
Dept. of ECE, Iowa State University.
Email: snt@iastate.edu
}
}

\maketitle

\doublespacing

\begin{abstract} 
We consider mining dense substructures (maximal cliques) from an uncertain graph, which is a probability distribution on a set of deterministic graphs. For parameter $0 < \alpha < 1$, we consider the notion of an $\alpha$-maximal clique in an uncertain graph. We present matching upper and lower bounds on the number of $\alpha$-maximal cliques possible within a (uncertain) graph. We present an algorithm to enumerate $\alpha$-maximal cliques whose worst-case runtime is near-optimal, and an experimental evaluation showing the practical utility of the algorithm.
\end{abstract}

\section{Introduction}
\label{sec:intro}
Large datasets often contain information that is uncertain in nature. For example, given people $A$ and $B$, it may not be possible to definitively assert a relation of the form ``$A$ knows $B$'' using available information. Our confidence in such relations are commonly quantified using probability, and we say that the relation exists with a probability of $p$, for some value $p$ determined from the available information. In this work, we focus on {\em uncertain graphs}, where our knowledge is represented as a graph, and there is uncertainty in the presence of each edge in the graph. Uncertain graphs have been used extensively in modeling, for example, in communication networks~\cite{GNSC2007,BM2005,KTMN2005}, social networks~\cite{AR2007,GKRT2004,KKT2003,LK2003,KG2010,BBGT2012}, protein interaction networks~\cite{AKGR2004,BCRC2004,RTVMBKGPC2005}, and regulatory networks in biological systems~\cite{JTCS2006}.

Identification of dense substructures within a graph is a fundamental task, with numerous applications in data mining, including in clustering and community detection in social and biological networks~\cite{PDFV2005}, the study of the co-expression of genes under stress~\cite{Rokhlenko07}, integrating different types of genome mapping data~\cite{Harley01}. Perhaps the most elementary dense substructure in a graph, also probably the most commonly used, is a clique, a completely connected subgraph. Typically, we are interested in a {\em maximal clique}, which is a clique that is not contained within any other clique. Enumerating all maximal cliques from a graph is one of the most basic problems in graph mining, and has been used in many settings, including in finding overlapping communities from social networks~\cite{PDFV2005,ML2014,BKS1980,PYB2012}, finding overlapping multiple protein complexes~\cite{Gavin02}, analysis of email networks~\cite{PMS2006} and other problems in bioinformatics~\cite{HBG2001,GARW1993,ZPKS2008}.

While the notion of a dense substructure, including that of a maximal clique, as well as methods for enumerating them, are well understood in a deterministic graph, the same is not true in the case of an uncertain graph. This is an important open problem today, given that many datasets increasingly incorporate data that is noisy and uncertain in nature. Uncertainty can result from a lack of data. For example, in constructing a social network from data collected through sensors, some communications between individuals maybe missed, or maybe anonymized~\cite{AR2007}. In some cases, relationships themselves are probabilistic in nature; for example, the relation of one person influencing another in a social network~\cite{CWY2009}. In biological networks such as protein--protein interaction networks, it is known that there are frequent errors in finding interactions and our knowledge is best modeled probabilistically~\cite{AKGR2004}.

In this work, we consider the analog of a maximal clique in an uncertain graph. Intuitively, a clique in an uncertain graph is a set of vertices that has a high probability of being a completely connected subgrap. In other words, when we sample from the uncertain graph, this set is likely to form a (deterministic) clique. Finding such sets of vertices enables us to unearth robust communities within an uncertain graph, for example, a group of proteins such that it is likely that each protein interacts with each other protein. We present a systematic study of the problem of identifying such structures within an uncertain graph.

\subsection{Our Contributions}
First, we present a precise definition of a maximal clique in an uncertain graph, leading to the notion of an $\alpha$-maximal clique, for parameter $0 < \alpha \le 1$. A set of vertices $U$ in an uncertain graph is an $\alpha$-maximal clique if $U$ is a clique with probability at least $\alpha$, and there does not exist a vertex set $U'$ such that $U \subset U'$ and $U'$ is a clique with probability at least $\alpha$. When $\alpha=1$, the above definition reduces to the well understood notion of a maximal clique in a deterministic graph.

\paragraph{Number of Maximal Cliques}
We first consider a basic question on maximal cliques in an uncertain graph: {\em how many $\alpha$-maximal cliques can be present within an uncertain graph?} For deterministic graphs, this question was first considered by Moon and Moser~\cite{MM65} in 1965, who presented matching upper and lower bounds for the largest number of maximal cliques within a graph; on a graph with $n$ vertices, the largest possible number of maximal cliques is $3^{\frac{n}{3}}$\footnote{This assumes that $3$ divides $n$. If not, the expressions are slightly different}. For the case of uncertain graphs, we present the first matching upper and lower bounds for the largest number of $\alpha$-maximal cliques in a graph on $n$ vertices. We show that for any $0 < \alpha < 1$, the maximum number of $\alpha$-maximal cliques possible in an uncertain graph is $n \choose \lfloor n/2 \rfloor$, i.e.  there is an uncertain graph on $n$ vertices with $n \choose \lfloor n/2 \rfloor$ uncertain maximal cliques and no uncertain graph on $n$ vertices can have more than $n \choose \lfloor n/2 \rfloor$ $\alpha$-maximal cliques.

\paragraph{Algorithm for Enumerating Maximal Cliques}
We present a novel algorithm, {\em MULE} ({\em M}aximal {\em U}ncertain c{\em L}ique {\em E}numeration), for enumerating all $\alpha$-maximal cliques within an uncertain graph. MULE is based on a depth-first-search of the graph, combined with optimizations for limiting exploration of the search space, and a fast way to check for maximality based on an incremental computation of clique probabilities. We present a theoretical analysis showing that the worst-case runtime of MULE is $O(n \cdot 2^{n})$, where $n$ is the number of vertices. This is nearly the best possible dependence on $n$, since our analysis of the number of maximal cliques shows that the size of the output can be as much as $O(\sqrt{n} \cdot 2^{n})$. Note that such worst-case behavior occurs only in graphs that are very dense; for typical graphs, we can expect the runtime of MULE to be far better, as we show in our experimental evaluation. We also present an extension of MULE to efficiently enumerate only large maximal cliques. 

\paragraph{Experimental Evaluation}
We present an experimental evaluation of MULE using synthetic as well as real-world uncertain graphs. Our evaluation shows that MULE is practical and can enumerate maximal cliques in an uncertain graph with tens of thousands of vertices, more than hundred thousand edges and more than two million $\alpha$-maximal cliques. Interestingly, the observed runtime of this algorithm is proportional to the size of the output. The real-world graphs included a protein--protein interaction network, and a collaboration network inferred from DBLP. 

\subsection{Related Work}
\label{sec:related}
There has been much recent work in the database and data mining communities on mining from uncertain graphs, including shortest paths~\cite{YCW2010}, nearest neighbors~\cite{PBGK2010}, clustering~\cite{KPT2013}, enumerating frequent and reliable subgraphs~\cite{HT2008,ZLGZ2010,JLA2011,ZGL2010,LRAY2012,KBGG2014},
and distance-constrained reachability~\cite{JLDW2011}. Our problem of enumerating dense substructures is different from the problems mentioned above. In particular, the problem of finding reliable subgraphs is one of finding subgraphs that are connected with a high probability. However, these individual subgraphs may be sparse. In contrast, we are interested in finding subgraphs that are not just connected, but also fully connected with a high probability. 
The most closely related work to ours is on mining cliques from an uncertain graph by Zou et. al~\cite{ZJHS2010}. Our work is different from theirs in significant ways as elaborated below.
\begin{itemize}
\item While we focus on enumerating all $\alpha$-maximal cliques in a graph, they focus on a different problem, that of enumerating the $k$ cliques with the highest probability of existence.

\item We present bounds on the number of such cliques that could exist, while by definition, their problem requires them to output no more than $k$ cliques.

\item We provide a runtime complexity analysis of our algorithm and show that it is near optimal. No runtime complexity analysis was provided for the algorithm presented in~\cite{ZJHS2010}.

\item We also provide an algorithm to enumerate only large maximal uncertain cliques.
\end{itemize}

There is substantial prior work on maximal clique enumeration from a deterministic graph. A popular algorithm for maximal clique enumeration problem is the Bron-Kerbosch algorithm~\cite{BK1973}, based on depth-first-search. Tomita et al.~\cite{TTT06} improved the depth-first-search approach through a better strategy for pivot selection; their resulting algorithm runs in time $O(3^{\frac{n}3})$, which is worst-case optimal, due to the bound on the number of maximal cliques possible~\cite{MM65}. Further work on enumeration of maximal cliques includes~\cite{Cazals08,Eppstein11,Modani08,Tsukiyama77,Chiba85,Johnson88,Makino04}.

{\bf Roadmap.} We present a problem definition in Section~\ref{sec:problem} and bounds on the number of $\alpha$-maximal cliques in Section~\ref{sec:number}. We present an algorithm to enumerate all $\alpha$-maximal cliques in Section~\ref{sec:algorithms}, followed by experimental results in Section~\ref{sec:expts}.

\remove{
{\color{blue}
Hence managing uncertain datasets have been a recent focus of the
Database community~\cite{AJKO2008,DHY2009,FHOR2011}.
For this work we focus on graph datasets. 
Graphs represent relationships which can sometimes only be
quantified probabilistically.
}
}

\remove{
{\color{blue}
Although dense substructures like quasi--clique, k-core etc can be
more practical due to the fact that it doesn't need all the edges 
to be present among the set of participating vertices~\cite{PYB2012}, 
these structures are not as fundamental as a maximal clique. 
Since even this fundamental problem is not solved in the context of 
uncertain graphs, we focus on enumerating maximal cliques in this 
work and leave other dense substructures as important open problems.
}
}

\section{Problem Definition}
\label{sec:problem}

An uncertain graph is a probability distribution over a set of
deterministic graphs. 
We deal with undirected simple graphs, i.e. there are no self-loops or
multiple edges. An uncertain graph is a triple $\ugraph=(V,E,p)$, 
where $V$ is a set of vertices, $E \subseteq V \times V$ is
a set of (possible) edges, and $p:E \rightarrow(0,1]$ is a function
that assigns a probability of existence $p(e)$ to each edge $e \in E$.
As in prior work on uncertain graphs, we assume that the existence of
different edges are mutually independent events.

Let $n=|V|$ and $m=|E|$. Note that $\mathcal{G}$ is a
distribution over $2^m$ deterministic graphs, each of which is a
subgraph of the undirected graph $(V,E)$. This set of possible 
deterministic graphs is
called the set of ``possible graphs'' of the uncertain graph
$\ugraph$, and is denoted by $D(\ugraph)$.  Note that in order to
sample from an uncertain graph $\ugraph$, it is sufficient to sample
each edge $e \in E$ independently with a probability $p(e)$.

In an uncertain graph $\ugraph=(V,E,p)$, two vertices $u$ and $v$ are
said to be adjacent if there exists an edge $\{u,v\}$ in $E$. Let the
neighborhood of vertex $u$, denoted $\Gamma(u)$, be the set of all
vertices that are adjacent to $u$ in $\ugraph$. The next two 
definitions are standard, and apply not to uncertain graphs, but to
deterministic graphs.

\begin{definition}
\label{def:CliqueInPossibleGraph} 
A set of vertices $C \subseteq V$ is a clique in a graph $G=(V,E)$, if
every pair of vertices in $C$ is connected by an edge in $E$.
\end{definition}

\begin{definition}
\label{def:MaxCliqueInPossibleGraph} 
A set of vertices $M \subseteq V$ is a maximal clique in a graph $G=(V,E)$, if
(1)~$M$ is a clique in $G$ and 
(2)~There is no vertex $v \in V \setminus M$ such that $M \cup \{v\}$ is a clique in $G$.
\end{definition}

\begin{definition}
In an uncertain graph $\ugraph$, for a set of vertices $C \subseteq V$,
the clique probability of $C$, denoted by $\cliqueprob(C,\ugraph)$, is
defined as the probability that in a graph sampled from $\ugraph$, $C$
is a clique. For parameter $0 \le \alpha \le 1$, $C$ is called an 
$\alpha$-clique if $\cliqueprob(C,\ugraph) \ge \alpha$.
\end{definition}

For any set of vertices $C \subseteq V$, let $E_C$ denote the set of
edges $\{e=\{u,v\} | e \in E, u,v \in C \mbox{ and } u \neq v\}$, i.e.
the set of edges connecting vertices in $C$.

\begin{observation}
\label{obs:Prod}
For any set of vertices $C \subseteq V$ in $\ugraph=(V,E,p)$, 
such that $C$ is a clique in $G=(V,E)$,
$\cliqueprob(C,\ugraph) = \prod_{\substack{e \in E_C}} p(e)$.
\end{observation}

\begin{proof}
Let $G$ be a graph sampled from $\ugraph$. The set $C$ will be
a clique in $G$ iff every edge in $E_C$ is present in $G$. 
Since the events of selecting different edges are independent
of each other, the observation follows. 
\end{proof}

\begin{definition}
\label{def:MaxCliqueUncertainGraph2}
Given an uncertain graph $\ugraph=(V,E,p)$, and a parameter 
$0 \le \alpha \le 1$, a set $M \subseteq V$ is defined as an 
$\alpha$-maximal clique if 
(1)~$M$ is an $\alpha$-clique in $\ugraph$, and 
(2)~There is no vertex $v \in (V \setminus M)$ such that $M \cup \{v\}$ is an $\alpha$-clique in $\ugraph$.
\end{definition}

\begin{definition}
\label{def:MaxCliqueEnumProblem}
The {\bf Maximal Clique Enumeration} problem in an Uncertain Graph
$\ugraph$ is to enumerate all vertex sets $M \subseteq V$ such that
$M$ is an $\alpha$-maximal clique in $\ugraph$.
\end{definition}

The following two observations follow directly from Observation~\ref{obs:Prod}.
\begin{observation}
\label{obs:Sub}
For any two vertex sets $A,B$ in $\ugraph$, if $B \subset A$ then, 
$\cliqueprob(B,\ugraph) \ge \cliqueprob(A,\ugraph)$.
\end{observation}

\begin{observation}
\label{obs:Elim}
Let $C$ be an $\alpha$-clique in $\ugraph$. Then for all $e \in E_C$ 
we have $p(e) \geq \alpha$.
\end{observation}


\section{Number of Maximal Cliques}
\label{sec:number}

The maximum number of maximal cliques in a deterministic graph on $n$
vertices is known exactly due to a result by Moon and Moser~\cite{MM65}. 
If $n \mod 3 = 0$, this number is $3^{\frac{n}{3}}$. If $n \mod 3 = 1$, then it 
is $4 \cdot 3^{\frac{n-4}{3}}$, and if $n \mod 3 = 2$, then it is $2 \cdot
3^{\frac{n-2}{3}}$. The graphs that have the maximum number of maximal
cliques are known as Moon-Moser graphs.

For uncertain cliques, no such bound was known so far. In this
section, we establish a bound on the maximum number of
$\alpha$-maximal cliques in an uncertain graph. 
For $0 < \alpha < 1$,
let $f(n,\alpha)$ be the maximum number of $\alpha$-maximal cliques in
any uncertain graph with $n$ nodes, without any assumption about the 
assignments of edge probabilities. 
The following theorem is the main result of this section.

\begin{theorem}\label{thm:maximum}
Let $n \ge 2$, and $0 < \alpha < 1$. Then: $f(n,\alpha)$ $=$ $n \choose {\lfloor n/2\rfloor}$
\end{theorem}

\begin{proof}
We can easily verify that the theorem holds for $n=2$.
for $n \ge 3$, let $g(n) = {n \choose {\lfloor n/2\rfloor}}$. 
We show $f(n,\alpha)$ is at least $g(n)$ in Lemma~\ref{lem:maxCliq}, and then show that
$f(n,\alpha)$ is no more than $g(n)$ in Lemma~\ref{lem:UpperBound}.
\end{proof}

\begin{lemma}
\label{lem:maxCliq}
For any $n \ge 3$, and any $\alpha, 0 < \alpha < 1$, there exists an
uncertain graph $\ugraph = (V,E,p)$ with $n$ nodes which has $g(n)$
$\alpha$-maximal cliques.
\end{lemma}

\begin{proof}
First, we assume that $n$ is even. 
Consider $\ugraph = (V,E,p)$, where
$E = V \times V$. Let $\mbox{$\kappa = {{n/2} \choose 2}$}$.
For each $e \in E$, let $p(e)=q$ where $q^{\kappa}=\alpha$.
We have $0 <q <1$ since $0< \alpha <1$.
Let $S$ be an arbitrary subset of $V$ such that  $|S|=n/2$.
We can verify that $S$ is an $\alpha$-maximal clique since
(1)~the probability that $S$ is a clique is $q^{\kappa}=\alpha$ and 
(2)~for any set $S^\prime \supsetneq S, S^\prime \subseteq V$,
the probability that $S^\prime$ is a clique is at most $q q^{\kappa}=q\alpha <\alpha$.
We can also observe that for any subset $S \subseteq V$, $S$ cannot be an
$\alpha$-maximal clique if $|S|< n/2$ or $|S|> n/2$. Thus
we conclude that a subset $S\subseteq V$ is an $\alpha$-maximal clique iff
$|S|=n/2$ which implies that the total number of $\alpha$-maximal
cliques in $\ugraph$ is $n \choose {n/2}$.
A similar proof applies when $n$ is odd. 
\end{proof}

Note that our construction in the Lemma above employs the condition
that $n \ge 3$ and $0<\alpha <1$. When $\alpha = 1$, the upper bound
is from the result of Moon and Moser for deterministic graphs, and in this
case $f(n,\alpha) = 3^{\frac{n}{3}}$ and is smaller than $g(n)$.
Next we present a useful definition required for proving the next Lemma.

\begin{definition}
\label{defn:non-redundant}
A collection of sets $\mathcal{C}$
is said to be non-redundant if for any pair $S_1, S_2 \in \mathcal{C}$,
$S_1 \ne S_2$, we have $S_1 \nsubseteq S_2$ and $S_2 \nsubseteq S_1$.
\end{definition}

\begin{lemma}
\label{lem:UpperBound}
$g(n)$ is an upper bound on $f(n,\alpha)$.
\end{lemma}

\begin{proof}
Let $\mathcal{C}^{\alpha}(\mathcal{G})$
be the collection of all $\alpha$-maximal cliques in $\mathcal{G}$.
Note that by the definition of $\alpha$-maximal cliques, any
$\alpha$-maximal clique $S$ in $\ugraph$ can not be a proper subset of
any other $\alpha$-maximal clique in $\ugraph$.
Thus from Definition~\ref{defn:non-redundant}, for any uncertain graph 
$\mathcal{G}$, $\mathcal{C}^{\alpha}(\mathcal{G})$ is a non-redundant collection.
Hence, it is clear that the largest number of $\alpha$-maximal cliques in $\ugraph$ 
should be upper bounded by the size of a largest non-redundant collection of subsets of $V$.

Let $\mathcal{C}$ be the collection of all subsets of $V$.
Based on $\mathcal{C}$, we construct such an undirected graph
$\widehat{G}=(\mathcal{C},\widehat{E})$
where for any two nodes $S_1 \in \mathcal{C}, S_2 \in \mathcal{C} $,
there is an edge connecting $S_1$ and $S_2$ iff $S_1 \subseteq S_2$ or $S_2 \subseteq S_1$.
It can be verified that a sub-collection $\mathcal{C}^\prime \subseteq \mathcal{C}$
is a non-redundant iff $\mathcal{C}^\prime$ is an independent set in $\widehat{G}$.
In Lemma~\ref{lem:maxiindep}, we show that $g(n)$
is the size of a largest independent set of $\widehat{G}$,
which implies that $g(n)$ is an upper bound for the number of 
$\alpha$-maximal cliques in $\ugraph$. 
\end{proof}

Let $\mathcal{C}^*$ be a largest independent set in $\widehat{G}$.
Also, let $ \mathcal{C}_k\subseteq \mathcal{C} , 0 \le k \le n$ be the collection
of subsets of $V$ with the size of $k$. Observe that for each $0 \le k \le n$,
$\mathcal{C}_k$ is an independent set of $\widehat{G}$.
Also let $L(n)$ and $U(n)$ be respectively the minimum and maximum size of sets in $\mathcal{C}^*$.
We can show that $L(n)$ and $U(n)$ can be bounded as shown in Lemma~\ref{lem:Lowe-n} and Lemma~\ref{lem:Upper-n} respectively. 


\begin{lemma}\label{lem:maxiindep}
For any $n \ge 3$, $\left | \mathcal{C}^* \right | = g(n)$.
\end{lemma}

\begin{proof}
We first consider the case when $n$ is even.
By Lemmas~\ref{lem:Lowe-n} and~\ref{lem:Upper-n}, we know
$n/2 \le L(n) \le U(n) \le n/2$. Thus we have
$L(n)=U(n)=n/2$ which implies $ \mathcal{C}^*= \mathcal{C}^*_{n/2}$.
Recall that $ \mathcal{C}_k\subseteq \mathcal{C} , 0 \le k \le n$ is the collection
of subsets of $V$ with the size of $k$.

We have (1)~$\mathcal{C}^*= \mathcal{C}^*_{n/2} \subseteq \mathcal{C}_{n/2}$ and
(2)~$|\mathcal{C}^*| \ge |\mathcal{C}_{n/2}| $ since
$\mathcal{C}^*$ is a largest independent set of $\widehat{G}$.
Thus we conclude $\mathcal{C}^*=\mathcal{C}_{n/2}$ which has the size
of ${n \choose (n/2)} = g(n)$.

We next consider the case when $n$ is odd. From
Lemmas~\ref{lem:Lowe-n} and~\ref{lem:Upper-n}, we know
$(n-1)/2 \le L(n) \le U(n) \le (n+1)/2$. Thus we have
$ \mathcal{C}^*= \mathcal{C}^*_{(n-1)/2} \bigcup \mathcal{C}^*_{(n+1)/2}$.
For notation convenience, we set $n_1=(n-1)/2, n_2=(n+1)/2$.
Let $\widehat{G}(\mathcal{C}_{n_1}, \mathcal{C}_{n_2})$ be the subgraph of $\widehat{G}$
induced by $\mathcal{C}_{n_1} \cup \mathcal{C}_{n_2}$. We can view
$\widehat{G}(\mathcal{C}_{n_1}, \mathcal{C}_{n_2})$ as a bipartite graph with two
disjoint vertex sets $\mathcal{C}_{n_1}$ and $\mathcal{C}_{n_2}$ respectively.
Observe that $\mathcal{C}^*_{n_1} \subseteq \mathcal{C}_{n_1}$
 and $\mathcal{C}^*_{n_2} \subseteq\mathcal{C}_{n_2}$.
Let $\widehat{E} (\mathcal{C}^*_{n_1})$
be the set of edges induced by $\mathcal{C}^*_{n_1}$ in
$\widehat{G}(\mathcal{C}_{n_1}, \mathcal{C}_{n_2})$. Since
$\mathcal{C}^*$ is an independent set of $\widehat{G}$,
none of the edges in $\widehat{E} (\mathcal{C}^*_{n_1})$
will have an end in a node of $\mathcal{C}^*_{n_2}$, i.e,
all the edges of $\widehat{E} (\mathcal{C}^*_{n_1})$ should have
an end  falling in $\mathcal{C}_{n_2} \setminus \mathcal{C}^*_{n_2}$.
Note that in $\widehat{G}(\mathcal{C}_{n_1}, \mathcal{C}_{n_2})$, all nodes
have a degree of $n_2$. Thus we have:
$$
|\widehat{E} (\mathcal{C}^*_{n_1})|=|\mathcal{C}^*_{n_1}|*n_2 \le |\mathcal{C}_{n_2} \setminus \mathcal{C}^*_{n_2}|*n_2= (|\mathcal{C}_{n_2}|-|\mathcal{C}^*_{n_2}|)*n_2
$$
from which we obtain $|\mathcal{C}^*|=|\mathcal{C}^*_{n_1}|+|\mathcal{C}^*_{n_2}| \le |\mathcal{C}_{n_2}|$ $=$ ${n} \choose n_2$.
Note that $\mathcal{C}_{n_2}$ itself is an independent set of
$\widehat{G}$ with size $n \choose n_2$.
Thus we conclude that $|\mathcal{C}^*|$ $=$ ${n} \choose n_2$ $= g(n)$.
\end{proof}

\begin{lemma}\label{lem:Lowe-n}
$L(n) \ge \lfloor n/2\rfloor$
\end{lemma}

\begin{proof}
Let us assume $n$ is an even number. We prove by contradiction as follows.
Suppose $L(n)=\ell \le n/2-1$. Let $\mathcal{C}^*_k \subseteq \mathcal{C}^*, L(n)\le k \le U(n)$ be
the collection of all sets in $\mathcal{C}^*$ which has the size of $k$, i.e,
$\mathcal{C}^*_k =\{S \in \mathcal{C}^*| |S|=k\}$. In the following we construct
a new collection $\mathcal{C}_{new} \subseteq \mathcal{C} $ which proves to be an
independent set in $\widehat{G}$ with the size being strictly larger than $\mathcal{C}^*$.
For each $S \in \mathcal{C}^*_\ell$, we add to $\mathcal{C}^*$ all subsets of $V$ which has the form
as $S\cup \{i\}$ where $i \in V \setminus S$ and remove $S$ from $\mathcal{C}^*$ meanwhile.
Let $\mathcal{C}_{new} $  be the collection obtained after we process the same route for all $S \in  \mathcal{C}^*_\ell$.
Mathematically, we have: $\mathcal{C}_{new}=  \mathcal{C}_1 \bigcup \mathcal{C}_2$ where 
$\mathcal{C}_1=\bigcup_{S\in \mathcal{C}^*_\ell } \bigcup_{i \in V \setminus S} \{ S\cup \{i\} \}, \mathcal{C}_2=\mathcal{C}^* \setminus \mathcal{C}^*_\ell$. 
First we show $\mathcal{C}_{new} $ is an independent set of $\widehat{G}$. Arbitrarily choose two distinct
sets, say $S_1 \in \mathcal{C}_{new}, S_2 \in \mathcal{C}_{new}, S_1 \ne S_2$.
We check all the possible cases one by one:
\begin{itemize}
  \item $S_1 \in \mathcal{C}_1, S_2 \in \mathcal{C}_1$. We observe that $|S_1|=|S_2|=\ell+1$ and $S_1 \ne S_2$. Thus
  no inclusion relation could exist between $S_1$ and $S_2$.
  \item $S_1 \in \mathcal{C}_2, S_2 \in \mathcal{C}_2$. In this case no inclusion relation can exist between $S_1$ and $S_2$
  since $\mathcal{C}_2$ is an independent set of $\widehat{G}$.
  \item $S_1 \in \mathcal{C}_1, S_2 \in \mathcal{C}_2$. Since $\mathcal{C}^*_\ell$ is the collection
  of sets in $\mathcal{C}^*$ which has the smallest size $\ell$, we get that $|S_2| \ge \ell+1=|S_1|$.
 Therefore there is only one possible inclusion relation existing here, that is $S_1 \subset S_2$. Suppose
 $S_1= S_1^\prime \cup \{i_1\} \subset S_2 $ for some $S_1^\prime \in \mathcal{C}^*_\ell$. Thus we get that
 $S_1^\prime  \subset S_2$ which implies  $\mathcal{C}^*$ is not an independent set of $\widehat{G}$. Hence
 we conclude that no inclusion relation could exist between $S_1$ and $S_2$.
\end{itemize}

Summarizing the analysis above, we get that no inclusion relation could exist between $S_1$ and $S_2$
which yields $\mathcal{C}_{new} $ is an independent set of $\widehat{G}$.

Now we prove that $|\mathcal{C}_{new} |>|\mathcal{C}^* | $. Observe that
$\mathcal{C}_1$ and $\mathcal{C}_2$ are disjoint from each other; otherwise
 $\mathcal{C}^*$ is not an independent set. So we have $|\mathcal{C}_{new}|=|\mathcal{C}_1|+|\mathcal{C}_2|$.
Note that $|\mathcal{C}^*|=|\mathcal{C}^*_\ell|+|\mathcal{C}_2|$ since
$\mathcal{C}^*$ is the union of the two disjoint parts $\mathcal{C}^*_\ell$ and $\mathcal{C}_2$.
Therefore
$|\mathcal{C}_{new} |>|\mathcal{C}^* | $ is equivalent to $|\mathcal{C}_1 |>|\mathcal{C}^*_\ell | $.
Let $ \widehat{G}(\mathcal{C}^*_\ell, \mathcal{C}_1)$ be the
induced subgraph graph of $\widehat{G}$ by $\mathcal{C}^*_\ell \bigcup \mathcal{C}_1$.
Note that $ \widehat{G}(\mathcal{C}^*_\ell, \mathcal{C}_1)$ can be viewed as
a bipartite graph where the two disjoint vertex sets are $\mathcal{C}^*_\ell $ and $\mathcal{C}_1$
respectively.
In $ \widehat{G}(\mathcal{C}^*_\ell, \mathcal{C}_1)$
we observe that (1) for each node $S_1 \in \mathcal{C}^*_\ell$, its degree
$d(S_1)=n-\ell$; (2) for each node $S_2 \in \mathcal{C}_1$, its degree $d(S_2) \le \ell+1$.
Thus we get that $ |\widetilde{E}|=|\mathcal{C}^*_\ell| (n-\ell) \le |\mathcal{C}_1|(\ell+1) $.
According to our assumption we have $\ell \le n/2-1$.
Thus we have \\ $|\mathcal{C}^*_\ell|/|\mathcal{C}_1| \le (\ell+1)/ (n-\ell) \le (n/2)/(n/2+1)<1$,
 yielding $|\mathcal{C}^*_\ell|< |\mathcal{C}_1|$ which is equivalent to
$|\mathcal{C}^* | < |\mathcal{C}_{new} | $.

So far we have successfully constructed a new collection $\mathcal{C}_{new} \subseteq \mathcal{C}$ 
such that (1) it is an independent set of $\widehat{G}$ and (2) $ |\mathcal{C}_{new} |>|\mathcal{C}^* |$.
That contradicts with the fact that $\mathcal{C}^*$ is a largest independent set of $\widehat{G}$.
Thus our assumption $\ell \le n/2-1$ does not hold, which yields $\ell \ge n/2$.
For the case when $n$ is odd, we can process essentially the same analysis as above and get $\ell \ge (n-1)/2$. 
\end{proof}

\begin{lemma}\label{lem:Upper-n}
$U(n) \le  \lceil n/2 \rceil $
\end{lemma}

\begin{proof}
Let us assume $n$ is an even number.
Based on $\mathcal{C}^*$, we
construct a dual collection $\mathcal{C}^*_{dual}$ as follows:
Initialize $\mathcal{C}^*_{dual}$ as an empty collection.
For each $S\in \mathcal{C}^*$, we add $V \setminus S$ into $\mathcal{C}^*_{dual}$.
Mathematically, we have: $\mathcal{C}^*_{dual}=\bigcup_{S\in \mathcal{C}^*} \{V \setminus S\}$.
First we show $\mathcal{C}^*_{dual}$ is an independent set of $\widehat{G}$.
Arbitrarily choose two distinct sets,
say $V \setminus S_1 \in \mathcal{C}^*_{dual}, V \setminus S_2 \in \mathcal{C}^*_{dual}$, where
$S_1 \in \mathcal{C}^*, S_2 \in \mathcal{C}^*, S_1 \ne S_2$. Note that
$$V \setminus S_1 \subset  V \setminus S_2 \Leftrightarrow S_1 \supset S_2, V \setminus S_2 \subset  V \setminus S_1 \Leftrightarrow S_2 \supset S_1 $$

Thus we have that no inclusion relation could exist between $V \setminus S_1$ and $V \setminus S_2$ since
no inclusion relation exists between $S_1$ and $S_2$ resulting from the fact that
$\mathcal{C}^*$ is an independent set of $\widehat{G}$. So we get
$\mathcal{C}^*_{dual}$ is an independent set as well.

We can verify that $|\mathcal{C}^*_{dual}|=|\mathcal{C}^*|$. Therefore we can conclude
$\mathcal{C}^*_{dual}$ is a largest independent set of $\widehat{G}$.
By Lemma~\ref{lem:Lowe-n}, we get to know the minimum size of sets in
$\mathcal{C}^*_{dual}$ should be at least $n/2$, which yields the maximum size of
of sets in $\mathcal{C}^*$ should be at most $n/2$.
For the case when $n$ is odd, we can analyze essentially the same as above. 
\end{proof}

\section{Enumeration Algorithm}
\label{sec:algorithms}

In this section, we present {\bf MULE} (Maximal Uncertain cLique Enumeration), an algorithm for enumerating all $\alpha$-maximal cliques in an uncertain graph $\ugraph$, followed by a proof of correctness and an analysis of the runtime. We assume that $\ugraph$ has no edges $e$ such that $p(e) < \alpha$.  If there are any such edges, they can be pruned away without losing any $\alpha$-maximal cliques, using Observation~\ref{obs:Elim}.  Let the vertex identifiers in $\ugraph$ be $1,2,\ldots,n$. For clique $C$, let $\max(C)$ denote the largest vertex in $C$.  For ease of notation, let $\max(\emptyset) = 0$, and let $\cliqueprob(\emptyset,\ugraph) = 1$. 


\paragraph{Intuition}
We first describe a basic approach to enumeration using depth-first-search (DFS) with backtracking. The algorithm starts with a set of vertices $C$ (initialized to an empty set) that is an $\alpha$-clique and incrementally adds vertices to $C$, while retaining the property of $C$ being an $\alpha$-clique, until we can add no more vertices to $C$. At this point, we have an $\alpha$-maximal clique. Upon finding a clique that is $\alpha$-maximal, the algorithm backtracks to explore other possible vertices that can be used to extend $C$, until all possible search paths have been explored. To avoid exploring the same set $C$ more than once, we add vertices in increasing order of the vertex id. For instance, if $C$ was currently the vertex set $\{1,3,4\}$, we do not consider adding vertex $2$ to $C$, since the resulting clique $\{1,2,3,4\}$ will also be reached by the search path by adding vertices $1,2,3,4$ in that order. 

MULE improves over the above basic DFS approach in the following ways. First, given a current $\alpha$-clique $C$, the set of vertices that can be added to extend $C$ includes only those vertices that are already connected to every vertex within $C$. Instead of considering every vertex that is greater than $\max(C)$, it is more efficient to track these vertices as the recursive algorithm progresses -- this will save the effort of needing to check if a new vertex $v$ can actually be used to extend $C$. This leads us to incrementally track vertices that can still be used to extend $C$. 

Second, note that not all vertices that extend $C$ into a clique preserve the property of $C$ being an $\alpha$-clique. In particular, adding a new vertex $v$ to $C$ decreases the clique probability of $C$ by a factor equal to the product of the edge probabilities between $v$ and every vertex in $C$. So, in considering vertex $v$ for addition to $C$, we need to compute the factor by which the clique probability will fall. This computation can itself take $\Theta(n)$ time since the size of $C$ can be $\Theta(n)$, and there can be $\Theta(n)$ edges to consider in adding $v$. A key insight is to reduce this time to $O(1)$ by incrementally maintaining this factor for each vertex $v$ still under consideration. The recursive subproblem contains, in addition to current clique $C$, a set $I$ consisting of pairs $(u,r)$ such that $u > \max(C)$, $u$ can extend $C$ into an $\alpha$-clique, and adding $u$ will multiply the clique probability of $C$ by a factor of $r$. This set $I$ is incrementally maintained and supplied to further recursive calls.

Finally, there is the cost of checking maximality. Suppose that at a juncture in the algorithm we found that $I$ was empty, i.e. there are no more vertices greater than $\max(C)$ that can extend $C$ into an $\alpha$-clique. This does not yet mean that $C$ is an $\alpha$-maximal clique, since it is possible there are vertices less than $\max(C)$, but not in $C$, which can extend $C$ to an $\alpha$-maximal clique (note that such an $\alpha$-maximal clique will be found through a different search path). This means that we have to run another check to see if $C$ is an $\alpha$-maximal clique. Note that even checking if a set of vertices $C$ is an $\alpha$-maximal clique can be a $\Theta(n^2)$ operation, since there can be as many as $\Theta(n)$ vertices to be potentially added to $C$, and $\Theta(n^2)$ edge interactions to be considered. We reduce the time for searching such vertices by maintaining the set $X$ of vertices that can extend $C$, but will be explored in a different search path. By incrementally maintaining probabilities with vertices in $I$ and $X$, we can reduce the time for checking maximality of $C$ to $\Theta(n)$.

MULE incorporates the above ideas and is described in Algorithm~\ref{algo:main}.
\begin{algorithm}[ht]

\caption{MULE($\ugraph,\alpha$)}
\KwIn{$\ugraph \mbox{ is the input uncertain graph}$}
\KwIn{$\alpha, 0 < \alpha < 1$ is the user provided probability threshold}
\label{algo:main}

$\hat{I} \leftarrow \emptyset$ \;
\ForAll{$u \in V$}
{
    $\hat{I} \leftarrow \hat{I} \cup \{  (u,1) \}$
}

Enum-Uncertain-MC($\emptyset$, 1 ,$\hat{I}$, $\emptyset$) \;

\end{algorithm}
\begin{algorithm}[ht]
\caption{Enum-Uncertain-MC($C,q,I,X$)}
\KwIn{We assume $\ugraph$ and $\alpha$ are available as immutable global variables} 
\KwIn{$C \mbox{ is the current Uncertain Clique being processed}$}
\KwIn{$q = \cliqueprob(C,\ugraph)$, maintained incrementally}
\KwIn{$I \mbox{ is a set of all tuples} \left(u,r\right)$, such that $\forall (u,r) \in I$, $u > max(C)$, and $\cliqueprob(C \cup \{ u\},\ugraph) = q \cdot r \ge \alpha$, i.e. $C \cup \{ u \}$ is an $\alpha$-clique in $\ugraph$}
\KwIn{$X \mbox{ is a set of all tuples} \left(v,s\right)$, such that $\forall (v,s) \in X$, $v \not \in C$, $v < max(C)$, and $\cliqueprob(C \cup \{v\},\ugraph) = q \cdot s \ge \alpha$ , i.e. $C \cup \{ v \}$ is an $\alpha$-clique in $\ugraph$}
\label{algo:dfs}

        \If{$I = \emptyset$ and $X = \emptyset$} 
        {
            \label{algoline:check}
            Output $C$ as $\alpha$-maximal clique \label{algoline:out}  \;
            \Return
        }

	\ForAll{$ (u,r) \in I$ considered in increasing order of $u$}
	{
                \label{algoline:loop}
                $C' \leftarrow C \cup \{ u \}$ \tcp{Note $m = max(C') = u$} \label{algoline:genc} 
                $q' \leftarrow q \cdot r$  \tcp{$\cliqueprob(C \cup \{ v \},\ugraph)$}  \label{algoline:prob}
                $I'  \leftarrow GenerateI(C',q',I)$ \; \label{algoline:genI}
                $X'  \leftarrow GenerateX(C',q',X)$ \; \label{algoline:genX}
                Enum-Uncertain-MC($C',q',I',X')$ \; \label{algoline:recursion}
                $X \leftarrow X \cup \{ (u,r) \}$ \label{algoline:addX}
        }
\end{algorithm}

\begin{algorithm}[ht]
\caption{GenerateI($C',q',I$)}
\label{algo:generateI}
\KwIn{We assume $\ugraph$ and $\alpha$ are available as immutable global variables} 

  $m \leftarrow max(C')$, $I' \leftarrow \emptyset$, $S \leftarrow \emptyset$ \;

  \ForAll{$(u,r) \in I$}
  {
    $S \leftarrow S \cup \{ u \}$   
  }

  $S \leftarrow S \cap \{ \Gamma(m) \}$

  \ForAll{$(u,r) \in I$}
  {
    \If{$u > m$ and $u \in S$}
    {
        $\cliqueprob(C' \cup\{  u \},\ugraph) \leftarrow q' \cdot r \cdot p(\{u,m\})$

        \If{$(\cliqueprob(C' \cup\{  u \},\ugraph)) \ge \alpha$}
        {
          $u' \leftarrow u$ \;
          $r' \leftarrow r \cdot p(\{u,m\})$ \;
          $I' \leftarrow I' \cup  \{ (u',r') \}$ 
        }
    }
  }

  \Return I'

\end{algorithm}


\begin{algorithm}[ht]
\caption{GenerateX($C',q',X$)}
\label{algo:generateX}
\KwIn{We assume $\ugraph$ and $\alpha$ are available as immutable global variables} 

  $m \leftarrow max(C')$, $X' \leftarrow \emptyset$, $S \leftarrow \emptyset$ \;

  \ForAll{$(v,s) \in I$}
  {
    $S \leftarrow S \cup \{ v \}$   
  }

  $S \leftarrow S \cap \{ \Gamma(m) \}$

  \ForAll{$(v,s) \in X$}
  {
    \If{$v \in S$}
    {
      $\cliqueprob(C' \cup\{  v \},\ugraph) \leftarrow q' \cdot s \cdot p(\{v,m\})$

      \If{$(\cliqueprob(C' \cup\{  v \},\ugraph) \ge \alpha$}
      {
        $v' \leftarrow v$ \;
        $s' \leftarrow s \cdot p(\{v,m\})$ \;
        $X' \leftarrow X' \cup  \{ (v',s') \}$
      }
    }
  }

  \Return X'

\end{algorithm}

\subsection{Proof of Correctness}

In this section we prove the correctness of MULE.

\begin{theorem}
\label{thm:correctness}
MULE (Algorithm~\ref{algo:main}) enumerates all $\alpha$-maximal cliques from an input uncertain graph $\ugraph$.
\end{theorem}

\begin{proof}
To prove the theorem we need to show the following. First, if $C$ is a clique emitted by Algorithm~\ref{algo:main}, then $C$ must be an $\alpha$-maximal clique. Next, if $C$ is an $\alpha$-maximal clique, then it will be emitted by Algorithm~\ref{algo:main}. We prove them in Lemmas~\ref{lemma:correctness1}~and~\ref{lemma:correctness2} respectively. 
\end{proof}

Before proving Lemmas~\ref{lemma:correctness1}~and~\ref{lemma:correctness2}, we prove some properties of Algorithm~\ref{algo:dfs}.

\begin{lemma}
\label{obs:I}
When Algorithm~\ref{algo:dfs} is called with $C'$ in line~\ref{algoline:recursion}, $I'$ is a set of all tuples $(u'r')$, where $u' \in V$ and $0 < r' \le 1$, such that, $\forall (u',r') \in I'$ , $u' > max(C')$,  and $\cliqueprob(C' \cup \{ u' \},\ugraph) = q' \cdot r' \ge \alpha$, i.e. $C' \cup \{ u' \}$ is an $\alpha$-clique in $\ugraph$.
\end{lemma}

\begin{proof}
Let $u' \in V$ be a vertex such that
(1)~$u' > \max(C')$, and
(2)~$C' \cup \{ u' \}$ is an $\alpha$-clique in $\ugraph$.
We need to show that $(u',r') \in I'$ such that
$\cliqueprob(C' \cup \{ u' \},\ugraph) = q' \cdot r'$.

Let $C'$ be a clique being called by Enum-Uncertain-MC with $I'$. 
Note that each call of the method adds one vertex $u \in I$ to the current clique $C$ 
such $u > \max(C)$. Since the vertices are added in the lexicographical ordering,
there is an unique sequence of calls to the method Enum-Uncertain-MC
such that we reach a point in execution of Algorithm~\ref{algo:dfs} where
Enum-Uncertain-MC is called with $C'$. We call this sequence of calls as
Call-$0$, Call-$1$, $\ldots$, Call-$\left | C' \right |$. 
Also, let $C_i$ be the clique used by method Enum-Uncertain-MC during Call-$i$.

We prove by induction. First consider the base case. For that
consider the first call made to Algorithm~\ref{algo:dfs}, i.e. Call-$0$.
We know that $C$ is initialized as $\emptyset$. 
During the first call made, all vertices in $V$ satisfy conditions (1) and (2). 
This is because, first $max(\emptyset) = 0$. 
Second any single vertex can be considered as a clique with probability $1$.
$\hat{I}$ is initialized such that all $r$ in $\hat{I}$ are $1 \ge \alpha$.
Thus for all u such that $(u,r) \in \hat{I}$, $u > \max(C)$.
This proves the base case.

For the inductive step, consider a recursive call to the method 
Call-$i$ which calls Call-$(i+1)$.
For every case expect initialization, $I'$ is generated from $I$ by 
line~\ref{algoline:genI} of Algorithm~\ref{algo:dfs} which in turn
calls Algorithm~\ref{algo:generateI}. In Algorithm~\ref{algo:generateI},
only vertices in $I$ that are greater than $C'$ are added to $I'$.
Thus all vertices in $I$ that satisfy (1) are added to $I'$.
Next every vertex in $I$ is connected to $C$. We need to show that all vertices
in $I'$ are connected to $C'$. 
In line 4 of Algorithm~\ref{algo:generateI}, we prune out any vertex in $I$
that is not connected to $m = \max(C')$.
Assume that $u'$ extends $C$ such that
$\cliqueprob(C \cup \{ u' \},\ugraph) = r$. Now let $c = \{ C' \setminus C$ \}.
Note that $c$ is a single vertex. Also, assume $u' > c$. 
From line 4, we know that $q' \cdot r' \ge \alpha$
Also from line 6 of Algorithm~\ref{algo:generateI}, $r' = r \cdot p(\{c,u'\})$.
Now $\cliqueprob(C' \cup \{ u' \},\ugraph) = q' \cdot r \cdot p(\{c,u'\}) = q' \cdot r'$,
Now in line 8 of Algorithm~\ref{algo:generateI} we add $u'$ to $I'$ only if $r' \ge \alpha$
thus proving the inductive step. 
\end{proof}

The following observation follows from Lemma~\ref{obs:I}.
\begin{observation}
\label{obs:cliqueprob}
The input $C$ to Algorithm~\ref{algo:dfs} is an $\alpha$-clique.
\end{observation}


\begin{lemma}
\label{obs:X}
When Algorithm~\ref{algo:dfs} is called with $C'$ in line~\ref{algoline:recursion}, 
$X'$ is a set of all tuples $(v',s')$, where $v' \in V$ and $0 < s' \le 1$,
such that, $\forall (v',s') \in X'$, we have $v' \not \in C'$, $v' < \max(C')$, 
and $(\cliqueprob(C' \cup \{ v' \},\ugraph) = q' \cdot s') \ge \alpha$,
i.e. $C' \cup \{ v' \}$ is an $\alpha$-clique in $\ugraph$.
\end{lemma}

\begin{proof}
Let $m = \max(C')$ and $C = C' \setminus \{ m\}$.
Since Algorithm~\ref{algo:dfs} was called with $C'$, 
it must have been called with $C$. This is because the working clique is always 
extended by adding vertices from $I$, and from Lemma~\ref{obs:I},
$I$ only contains vertices that are greater than the maximum vertex in $C$.
Let $X$ be the corresponding set of tuples used when the call was made to 
Enum-Uncertain-MC with $C$. Let $u > \max(C)$ be a vertex such that 
$\cliqueprob(C' \cup \{ u \},\ugraph) \ge \alpha$ and $u < m$.
Note that $u \not \in C'$, $u < \max(C')$, and 
$C' \cup \{ u \}$ is an $\alpha$-clique in $\ugraph$. 
This means $u$ satisfies all conditions for $u \in X'$.
We need to show that when Enum-Uncertain-MC is called with $C'$,
the generated $X'$ which is passed in Enum-Uncertain-MC contains $u$.

Firstly, note that since $C' \cup \{ u \}$ is $\alpha$-clique in $\ugraph$,
we have $\cliqueprob(C \cup \{ u \},\ugraph) \ge \alpha$ (from Observation~\ref{obs:Sub}).
Since $u > \max(C)$ and $\cliqueprob(C \cup \{ u \},\ugraph) \ge \alpha$,
from Lemma~\ref{obs:I}, $u$ will be used in line~\ref{algoline:loop}
to call Enum-Uncertain-MC using $C \cup \{ u \}$.
Once this call is returned, $u$ is added to $X$ in line~\ref{algoline:addX}.
Note that since the loop at line~\ref{algoline:loop} add vertices in 
lexicographical order, $m$ will be added to $C$ after $u$.
Thus $u$ will be in $X$, when $m$ is used to extend $C$.
Next we show that if $u \in X$, after execution of line~\ref{algoline:genX},
$u \in X'$. We prove this as follows. Note that Algorithm~\ref{algo:generateX} 
is used to generate $X'$ from $X$. Note that $X'$ is generated by 
Algorithm~\ref{algo:generateX} by selectively adding vertices from $X$.
A vertex is added to $X'$ from $X$, only if $C' \cup \{ u \}$ is $\alpha$-clique in $\ugraph$. 
From our initial assumptions, we know that $u$ satisfies this condition and is 
hence added to $X'$ and passed on to Enum-Uncertain-MC when it is called with $C'$.

Now let us consider $v$, such that $v$ does not satisfy all the conditions
for $v \in X'$. We need to show that $v \not \in X'$. There are two cases.
First, when $v \not \in X$. This case is trivial as $X'$ is constructed from
$X$ and hence if $v \not \in X$, $v \not \in X'$.
For the second case, when $v \in X$, we need to show that $v$
will not be added to $X'$ in line~\ref{algoline:genX} of Algorithm~\ref{algo:dfs}. 
Note that since $v \in X$, we know $v \not \in C'$ and $v < \max(C')$.
Thus, it must be that $C \cup \{ m, v \}$ is not an $\alpha$-clique in $\ugraph$.
Algorithm~\ref{algo:generateX} will add $v$ to $X'$ only if
$C \cup \{ m, v \}$ is $\alpha$-clique in $\ugraph$. 
But from our previous discussion, we know that this condition doesn't hold.
Hence, $v$ will not be added to $X'$.
Thus only vertices that satisfy all three conditions are in $X'$.  
\end{proof}

\begin{lemma}
\label{lemma:correctness1}
Let $C$ be a clique emitted by Algorithm~\ref{algo:dfs}. Then $C$ is an $\alpha$-maximal clique. 
\end{lemma}

\begin{proof}
Algorithm~\ref{algo:dfs} emits $C$ in Line~\ref{algoline:out}.
From Observation~\ref{obs:cliqueprob}, we know that $C$ is an $\alpha$-clique. We need to show that $C$ is $\alpha$-maximal. We use proof by contradiction. Suppose $C$ is non-maximal. This means that there exists a vertex $u \in V$, such that  $C \cup \{u\}$ is an $\alpha$-clique. We know that $I = \emptyset$ when $C$ is emitted. From Lemma~\ref{obs:I}, we know that there exists no vertex $u \in V$ such that $u > max(C)$ that can extend $C$. Again, we know that $X = \emptyset$ when $C$ is emitted. Thus from Lemma~\ref{obs:X}, we know that there exists no vertex $v \in V$ such that $v < \max(C)$ that can extend $C$. This is a contradiction and hence $C$ is an $\alpha$-maximal clique. 
\end{proof}

\begin{lemma}
\label{lemma:correctness2}
Let $C$ be an $\alpha$-maximal clique in $\ugraph$. Then $C$ is emitted by Algorithm~\ref{algo:dfs}.
\end{lemma}

\begin{proof}
We first show that a call to method Enum-Uncertain-MC with $\alpha$-clique $C$ enumerates all $\alpha$-maximal cliques $C'$ in $\ugraph$, such that for all $c \in \{ C' \setminus C \}$, $c > \max(C)$.

Without loss of generality, consider a $\alpha$-maximal clique $C'$ in $\ugraph$ such that $\forall c \in \{C' \setminus C \}$, $c > \max(C)$. Note that $C'$ will be emitted as an $\alpha$-maximal clique by the method Enum-Uncertain-MC when called with $C$, if the following holds: (1) A call to method Enum-Uncertain-MC is made with $C'$, (2) When this call is made, $I' = \emptyset$, and $X' = \emptyset$. Since $C'$ is $\alpha$-maximal clique in $\ugraph$, the second point follows from Lemmas~\ref{obs:I}~and~\ref{obs:X}. Thus we need to show that a call to Enum-Uncertain-MC is made with $C'$.

We prove this by induction.
Let $\hat{C} = \{ C' \setminus C \}$.
Let $c_i$ represent the $i$th element in $\hat{C}$ in lexicographical order.
Also let $C_i = C \cup \{ c_{1}, c_{2}, \ldots, c_i \}$.
For the base case, we show that if a call to Enum-Uncertain-MC is made with $C$,
a call will be made with $C_1 = C \cup \{ c_1 \}$. This is because, 
line~\ref{algoline:loop} of the method loops over every vertex $u \in I$ thus
implying $u > \max(C)$ and $\cliqueprob(C \cup \{ u \},\ugraph) \ge \alpha$.
Since $C'$ is an $\alpha$-maximal clique, $c_1$ will satisfy both these
conditions and hence a call to Enum-Uncertain-MC is made with $C \cup \{ c_1 \}$.
Now for the inductive step we show that if a call is made with clique
$C_i$, then this call will in turn call the method with clique $C_{i+1}$.
Again, $c_{i+1}$ is greater than $\max(C_i)$ and 
$\cliqueprob(C_i \cup \{ c_{i+1} \},\ugraph) \ge \alpha$. 
Thus $c_{i+1} \in I$ when the call is made to Enum-Uncertain-MC with $C_i$.
Hence using the previous argument, in line~\ref{algoline:loop}, $c_{i+1}$ will be used as 
a vertex in the loop which would in turn make a call to Enum-Uncertain-MC with $C_{i+1}$. 

Now without any loss of generality, consider an $\alpha$-maximal clique in $\ugraph$.
We know that $C \supset \emptyset$. Thus the proof follows. 
\end{proof}

\subsection{Runtime Complexity}

\begin{theorem}
\label{thm:complexity}
The runtime of MULE (Algorithm~\ref{algo:main}) on an input graph of $n$ vertices is $O\left (n \cdot 2^n \right)$.
\end{theorem}

\begin{proof}
MULE initializes variables and calls to Algorithm~\ref{algo:dfs}, hence we analyze the runtime of Algorithm~\ref{algo:dfs}. An execution of the recursive Algorithm~\ref{algo:dfs} can be viewed as a search tree as follows. Each call to Enum-Uncertain-MC is a node of this search tree. The first call to the method is the root node. A node in this search tree is either an internal node that makes one or more recursive calls, or a leaf node that does not make further recursive calls. To analyze the runtime of Algorithm~\ref{algo:dfs}, we consider the time spent at internal nodes as well as leaf nodes.

The runtime at each leaf node is $O(1)$. For a leaf node, the parameter $I = \emptyset$, and there are no further recursive calls. This implies that either $C$ is $\alpha$-maximal ($X = \emptyset$) and is emitted in line~\ref{algoline:out} or it is non-maximal ($X \neq \emptyset$) but cannot be extended by the loop in line~\ref{algoline:loop} as $I = \emptyset$. Checking the sizes of $I$ and $X$ takes constant time.

We next consider the time taken at each internal node. Instead of adding up the times at different internal nodes, we equivalently add up the cost of the different edges in the search tree. At each internal node, the cost of making a recursive call can be analyzed as follows. Line~\ref{algoline:genc} takes $O \left ( n \right )$ time as we add all vertices in $C$ to $C'$ and also $u$. Line~\ref{algoline:prob} takes constant time. Lines~\ref{algoline:genI}~and~\ref{algoline:genX} take $O \left ( n \right )$ time (Lemmas~\ref{lemma:complexity1}~and~\ref{lemma:complexity2} respectively). Note that lines~\ref{algoline:genc}~to~\ref{algoline:genX} can get executed only once in between the two calls. Thus total runtime for each edge of the search tree is $O\left ( n \right )$.

Note that the total number of calls made to the method method Enum-Uncertain-MC is no more than the possible number of unique subsets of $V$, which is  $O \left ( 2^n \right )$. We see that for internal nodes, time complexity is $O \left ( n \right )$ and for leaf nodes it is $O \left ( 1 \right )$. Hence the time complexity of Algorithm~\ref{algo:dfs} is $O \left ( n \cdot 2^n \right )$. 
\end{proof}

Thus now we need to prove that lines~\ref{algoline:genI}~and~\ref{algoline:genX}
take  $O \left ( n \right )$ time. This implies that time complexity of
Algorithms~\ref{algo:generateI}~and~\ref{algo:generateX} is $O \left ( n \right )$.
We prove the same in Lemmas~\ref{lemma:complexity1}~and~\ref{lemma:complexity2}
respectively.

\begin{lemma}
\label{lemma:complexity1}
The runtime of Algorithm~\ref{algo:generateI} is $O \left ( n \right )$.
\end{lemma}

\begin{proof}
First note that lines 1-6 takes $O \left ( n \right )$ time. 
This is because $\left | I \right | = O \left ( n \right )$,
and hence the loop at line 4 of Algorithm~\ref{algo:generateI}
can take $O \left ( n \right )$ time. Further the set intersection
at line 6 also takes $O \left ( n \right )$ time.
We need to show that the for loop in line 7 is $O \left ( n \right )$,
that is each iteration of the loop takes $O \left ( 1 \right )$ time. 
Assume that it takes constant time to find out the probability of an edge.
This is a valid assumption, as the edge probabilities can be stored as a
HashMap and hence for an edge $e$, in constant time we can find out $p(e)$.
With this assumption, it is easy to show that lines 8-13 takes constant time.
This is because, they are either constant number of multiplications, or adding 
one element to a set. Thus total time complexity is $O \left ( n \right )$. 
\end{proof}

\begin{lemma}
\label{lemma:complexity2}
The runtime of Algorithm~\ref{algo:generateX} is $O \left ( n \right )$.
\end{lemma}

We omit the proof of the above lemma since it is similar to the proof
of Lemma~\ref{lemma:complexity1}.

\begin{observation}
\label{obs:opt}
The worst-case runtime of any algorithm that can output all maximal cliques of an uncertain graph on $n$ vertices is $\Omega \left ( \sqrt{n} \cdot 2^n \right )$.
\end{observation}

\begin{proof}
From Theorem~\ref{thm:maximum}, we know that the number of maximal uncertain cliques
can be as much as ${n \choose \lfloor n/2 \rfloor} = \Theta \left (\frac{{2^n}}{\sqrt{n}} \right )$
(using Stirling's Approximation). Since the size of each uncertain
clique can be $\Theta\left ( n \right )$, the total output size can be
$\Omega \left ( \sqrt{n} \cdot 2^n \right )$, which is a lower bound
on the runtime of any algorithm. 
\end{proof}

\begin{lemma}
\label{obs:timefactor}
The worst-case runtime of MULE on an $n$ vertex graph is within a $O(\sqrt{n})$ factor of the runtime of an optimal algorithm for Maximal Clique Enumeration on an uncertain graph.
\end{lemma}

\begin{proof}
The proof follows from Theorem~\ref{thm:complexity} and Observation~\ref{obs:opt}. 
\end{proof}

\subsection{Enumerating Only Large Maximal Cliques}
For a typical input graph, many maximal cliques are small, and may not be interesting to the user. Hence it is helpful to have an algorithm that can enumerate only large maximal cliques efficiently, rather than enumerate all maximal cliques. We now describe an algorithm that enumerates every $\alpha$-maximal clique with more than $t$ vertices, where $t$ is an user provided parameter.

As a first step, we prune the input uncertain graph $\ugraph = (V,E,p)$ by employing techniques described by Modani and Dey~\cite{Modani08}. We apply the ``Shared Neighborhood Filtering" where edges are recursively checked and removed as follows. First drop all edges $\{u,v\} \in E$, such that $\left | \Gamma(u) \cap \Gamma(v) \right | < (t-2)$. Next drop every vertex $v \in V$, that doesn't satisfy the following condition. For vertex $v \in V$, there must exist at least $(t-1)$ vertices in $\Gamma(v)$, such that for $u \in \Gamma(v)$, $\left | \Gamma(u) \cap \Gamma(v) \right | < (t-2)$. 
Let $\ugraph'$ denote the graph resulting from $\ugraph$ after the pruning step.

Algorithm~\ref{algo:largemule} runs on the pruned uncertain graph $\ugraph'$ to enumerate only large maximal cliques. The recursive method in Algorithm~\ref{algo:large} differs from Algorithm~\ref{algo:dfs} as follows. Before each recursive call to method Enum-Uncertain-MC-Large (Algorithm~\ref{algo:large}), the algorithm checks if the sum of the sizes of the current working clique $C'$ and the candidate vertex set $I'$ are greater than the size threshold $t$. If not, the recursive method is not called. This optimization leads to a substantial pruning of the search space and hence a reduction in runtime.

\begin{algorithm}[ht]
\caption{LARGE--MULE($\ugraph,\alpha$,$t$)}
\KwIn{$\ugraph' \mbox{ is the input uncertain graph post pruning}$}
\KwIn{$\alpha, 0 < \alpha < 1$ is the user provided probability threshold}
\KwIn{$t, t \ge 2$ is the user provided size threshold}
\label{algo:largemule}

$\hat{I} \leftarrow \emptyset$ \;
\ForAll{$u \in V$}
{
    $\hat{I} \leftarrow \hat{I} \cup \{  (u,1) \}$
}

Enum-Uncertain-MC-Large($\emptyset$, 1 ,$\hat{I}$, $\emptyset$,$t$) \;

\end{algorithm}
\begin{algorithm}[ht]
\caption{Enum-Uncertain-MC-Large($C,q,I,X$,$t$)}
\KwIn{$C \mbox{ is the current Uncertain Clique being processed}$}
\KwIn{$q \mbox{ is pre-computed } \cliqueprob(C,\ugraph)$}
\KwIn{$I \mbox{ is a set of tuples} \left(u,r\right)$, such that $\forall (u,r) \in I$, $u > \max(C)$, and $\cliqueprob(C \cup \{ u\},\ugraph) = q \cdot r \ge \alpha$, i.e. $C \cup \{ u \}$ is an $\alpha$-clique in $\ugraph$}
\KwIn{$X \mbox{ is a set of tuples} \left(v,s\right)$, such that $\forall (v,s) \in X$, $v \not \in C$, $v < \max(C)$, and $\cliqueprob(C \cup \{v\},\ugraph) = q \cdot s \ge \alpha$ , i.e. $C \cup \{ v \}$ is an $\alpha$-clique in $\ugraph$}
\KwIn{$t \mbox{ is the user provided size threshold}$}
\label{algo:large}

        \If{$I = \emptyset$ and $X = \emptyset$} 
        {
            Output $C$ as $\alpha$-maximal clique  \;
            \Return
        }

	\ForAll{$u,r \in I$ taken in lexicographical ordering of $u$}
	{
                $C' \leftarrow C \cup \{ u \}$ \tcp{Note $m = \max(C') = u$}  
                $q' \leftarrow q \cdot r$  \tcp{$\cliqueprob(C \cup \{ v \},\ugraph)$}  
                $I'  \leftarrow GenerateI(C',q',I)$ \; 

                \If{$\left | C' \right | + \left | I' \right | < t$}
                {
                    continue \;
                }

                $X'  \leftarrow GenerateX(C',q',X)$ \; 
                Enum-Uncertain-MC-Large($C',q',I',X',t)$ \; 
                $X \leftarrow X \cup \{ (u,r) \}$ 
        }
\end{algorithm}

\begin{lemma}
\label{lem:large}
Given an input graph $\ugraph$, LARGE--MULE (Algorithm~\ref{algo:largemule}) enumerates every $\alpha$-maximal clique with more than $t$ vertices.
\end{lemma}

\begin{proof}
First we prove that no maximal clique of size less than $t$ is enumerated by Algorithm~\ref{algo:large}. Consider an $\alpha$-maximal clique $C_1$ in $\ugraph$ with less than $t$ vertices. Also let $m_1 = \max(C_1)$ and $C_1' = C_1 \setminus \{ m_1 \}$.  Note that if $C_1$ is emitted by Algorithm~\ref{algo:large}, then a call must be made to Enum-Uncertain-MC-Large with $C_1$. Since the Algorithm adds vertices in lexicographical ordering, this implies that a call must be made to Enum-Uncertain-MC-Large with $C_1'$ before the call is made with $C_1$. In the worst case, let us consider that the search tree reaches the execution point where Enum-Uncertain-MC-Large is called with $C_1'$. Consider the execution of the algorithm where $m_1$ is added to $C = C_1'$ to form $C' = C_1$. Since $C_1$ is an $\alpha$-maximal clique, $I'$ will become NULL which implies $\left | I' \right | = 0$. We know that $\left | C_1 \right | < t$. Thus $\left | C_1 + I' \right |$ will also be less than $t$ and the If condition (line 8) will succeed. This will result in the execution of the continue statement. Thus Enum-Uncertain-MC-Large will not be called with $C_1$ implying that $C_1$ is not enumerated.

Next we show that any maximal clique of size at least $t$ is enumerated by Algorithm~\ref{algo:large}. Consider an $\alpha$-maximal clique $C_2$ in $\ugraph$ of size at least $t$. We note that the ``If'' condition in line 8 is never satisfied in the search path ending with $C_2$ and hence a call is made to the method with Enum-Uncertain-MC-Large with $C_2$. This is easy to see as whenever a call is made to Enum-Uncertain-MC-Large with any $C \subseteq C_2$, since $C_2$ is large, we always have $\left | C \right | + \left | I \right | \ge t$. 
\end{proof}

\section{Experimental Results}
\label{sec:expts}

We report the results of an experimental evaluation of our algorithm. We implemented the algorithm using Java. We ran all experiments on a system with a 3.19 GHz Intel(R) Core(TM) i5 processor and 4 GB of RAM, with heap space configured at 1.5GB.


%
\begin{table*}
\caption{Input Graphs}
\label{table:input}
\centering
\small
\begin{tabular}{c c c c c} 
\hline\hline 
Input Graph & Category & Description & \# Vertices & \# Edges \\ [0.5ex]
\hline
Fruit-Fly & Protein Protein Interaction network & PPI for Fruit Fly from STRING Database & 3751 & 3692 \\
DBLP10 & Social network & Collaboration network from DBLP & 684911 & 2284991 \\
p2p-Gnutella08 & Internet peer-to-peer networks & Gnutella network August 8 2002 & 6301 & 20777 \\
p2p-Gnutella04 & Internet peer-to-peer networks & Gnutella network August 4 2003 & 10879 & 39994 \\
p2p-Gnutella09 & Internet peer-to-peer networks & Gnutella network August 9 2003 & 8114 & 26013 \\
ca-GrQc & Collaboration networks & Arxiv General Relativity & 5242 & 28980 \\
wiki-vote & Social networks & wikipedia who-votes-whom network & 7118 & 103689 \\ 
BA5000 & Barab{\'a}si$-$Albert random graphs & Random graph with 5K vertices & 5000 & 50032 \\
BA6000 & Barab{\'a}si$-$Albert random graphs & Random graph with 6K vertices & 6000 & 60129 \\
BA7000 & Barab{\'a}si$-$Albert random graphs & Random graph with 7K vertices & 7000 & 70204 \\
BA8000 & Barab{\'a}si$-$Albert random graphs & Random graph with 8K vertices & 8000 & 80185 \\
BA9000 & Barab{\'a}si$-$Albert random graphs & Random graph with 9K vertices & 9000 & 90418 \\
BA10000 & Barab{\'a}si$-$Albert random graphs & Random graph with 10K vertices & 10000 & 99194 \\
\hline
\end{tabular}
\vspace{-0.5cm}
\end{table*}

\begin{figure*}[!ht]
\vspace{-0.5cm}
\begin{center}

$\begin{array}{cc}

\subfloat[$\alpha = 0.9$]{\label{fig:0a}\includegraphics[width=0.35\textwidth]{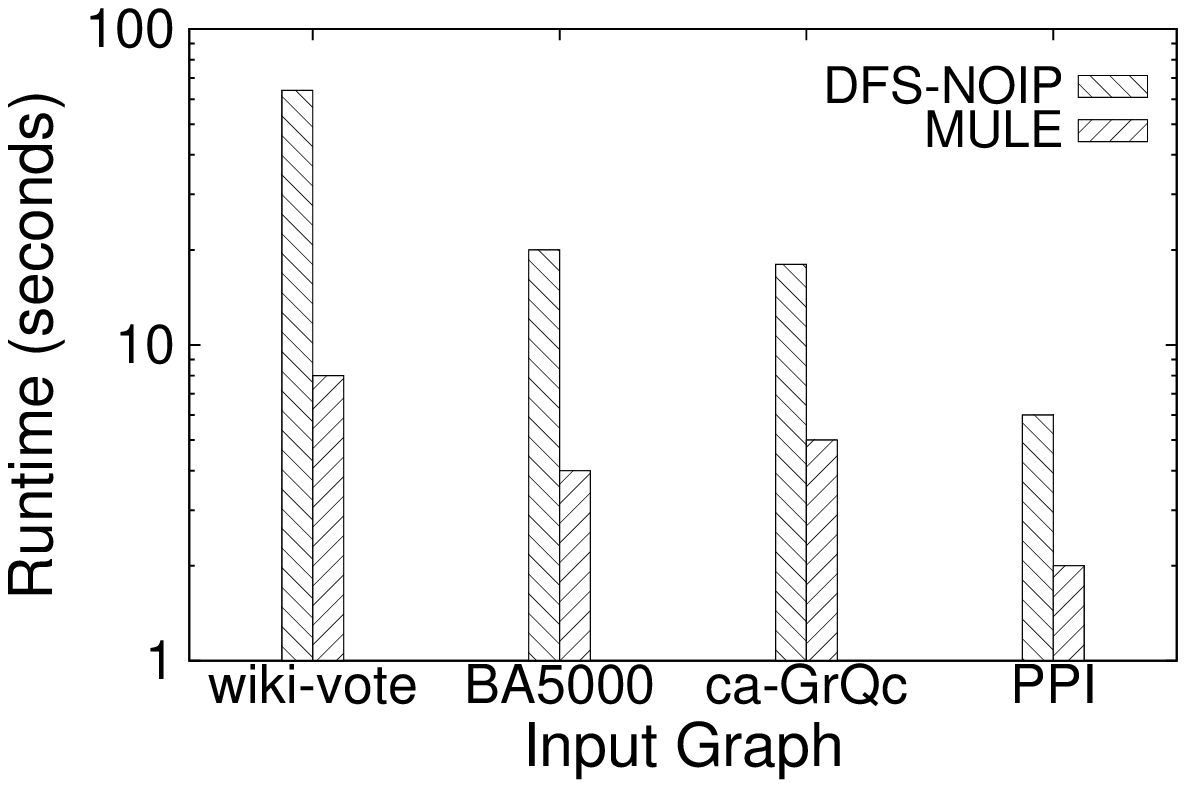}}  

\subfloat[$\alpha = 0.8$]{\label{fig:0b}\includegraphics[width=0.35\textwidth]{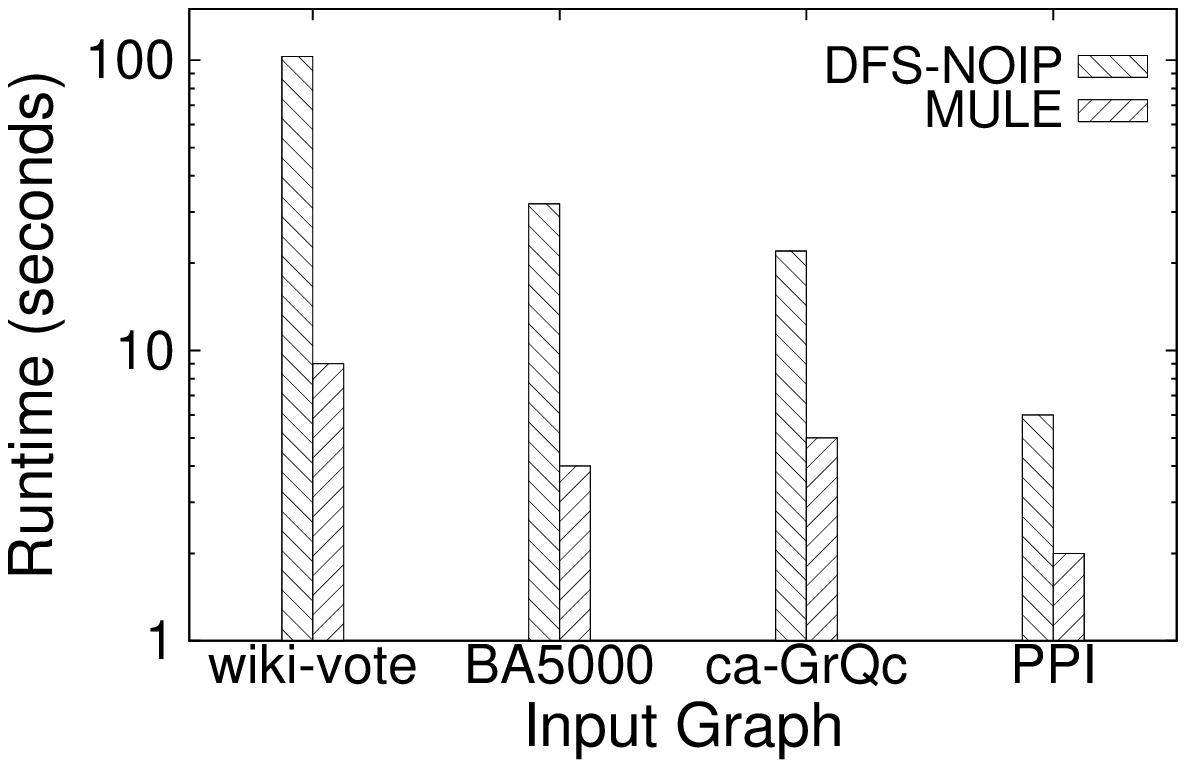}} 

\end{array}$

\vspace{-0.9cm}

$\begin{array}{cc}

\subfloat[$\alpha = 0.0001$]{\label{fig:0c}\includegraphics[width=0.35\textwidth]{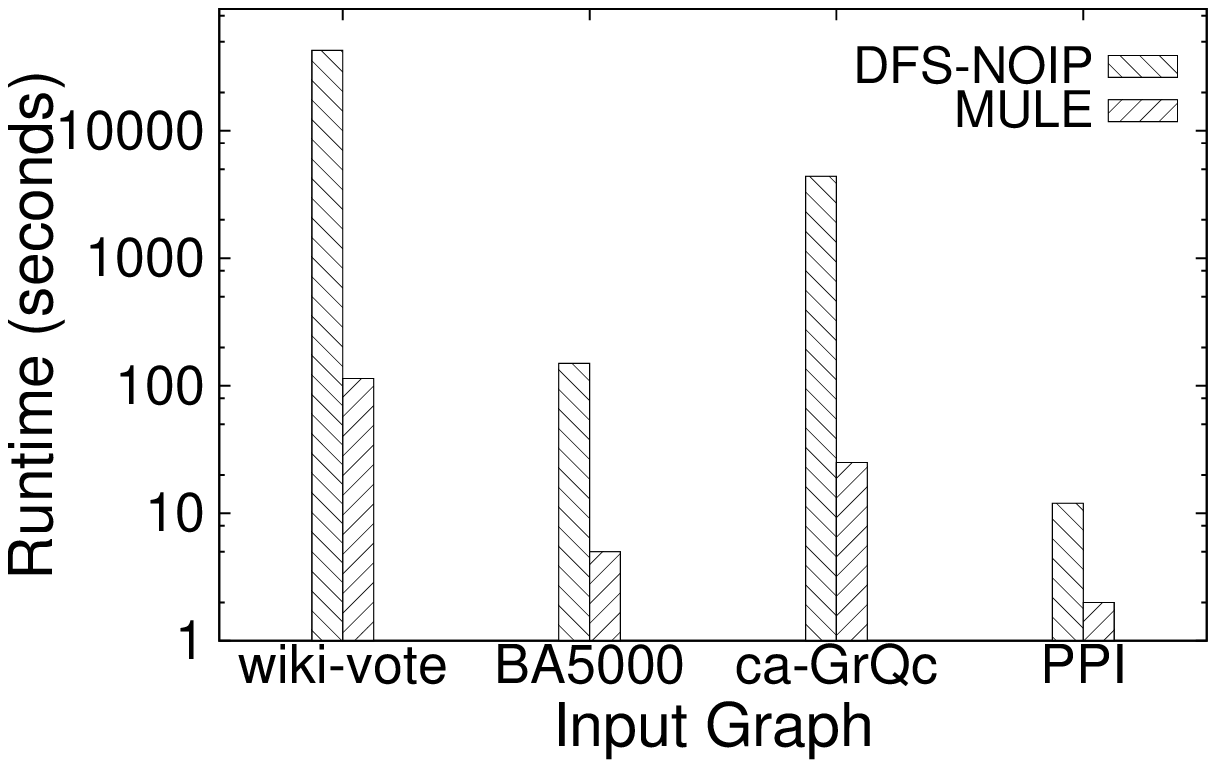}} 

\subfloat[$\alpha = 0.0005$]{\label{fig:0d}\includegraphics[width=0.35\textwidth]{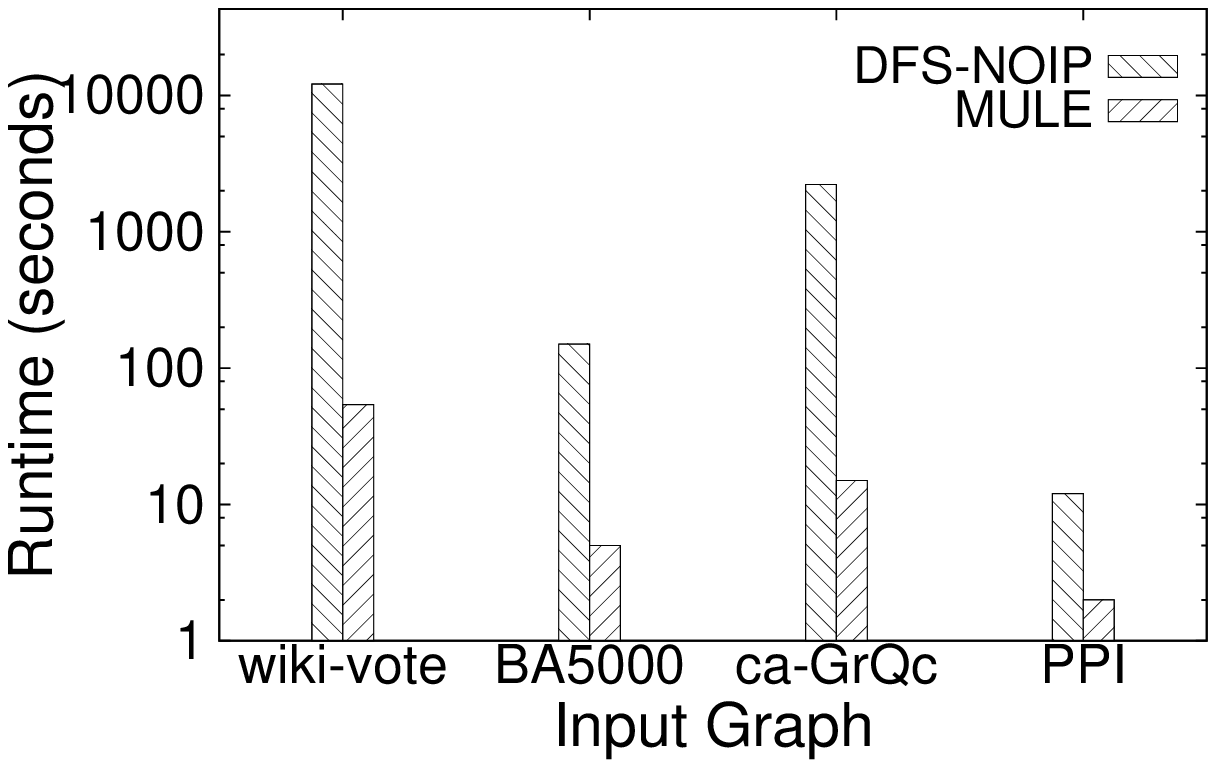}} 

\end{array}$

\end{center}
\vspace{-0.3cm}
\caption{Comparison of Simple and Optimized Depth First Search approaches. 
The Y--Axis is in log--scale.} 
\vspace{-0.5cm}
\label{fig:0}
\end{figure*}

\begin{figure*}[!ht]
\vspace{-0.5cm}
\begin{center}

$\begin{array}{cc}

\subfloat[Random Graphs]{\label{fig:1a}\includegraphics[width=0.35\textwidth]{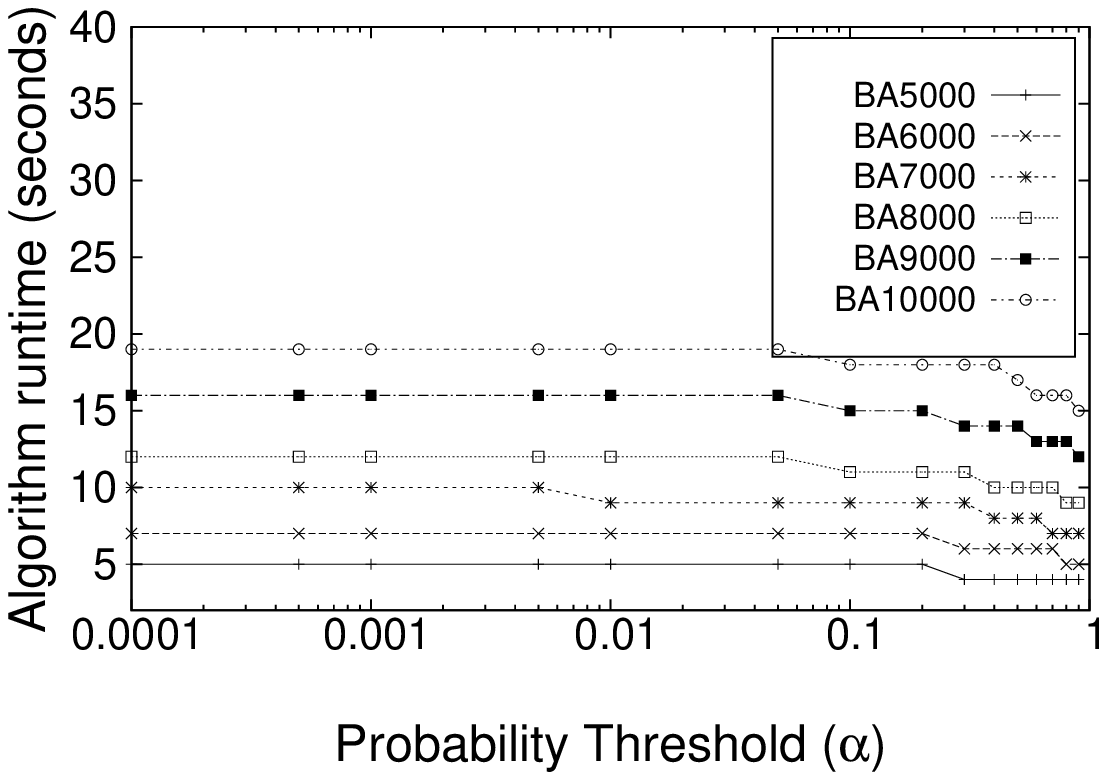}} &

\subfloat[Semi--synthetic and Real Graphs]{\label{fig:1b}\includegraphics[width=0.35\textwidth]{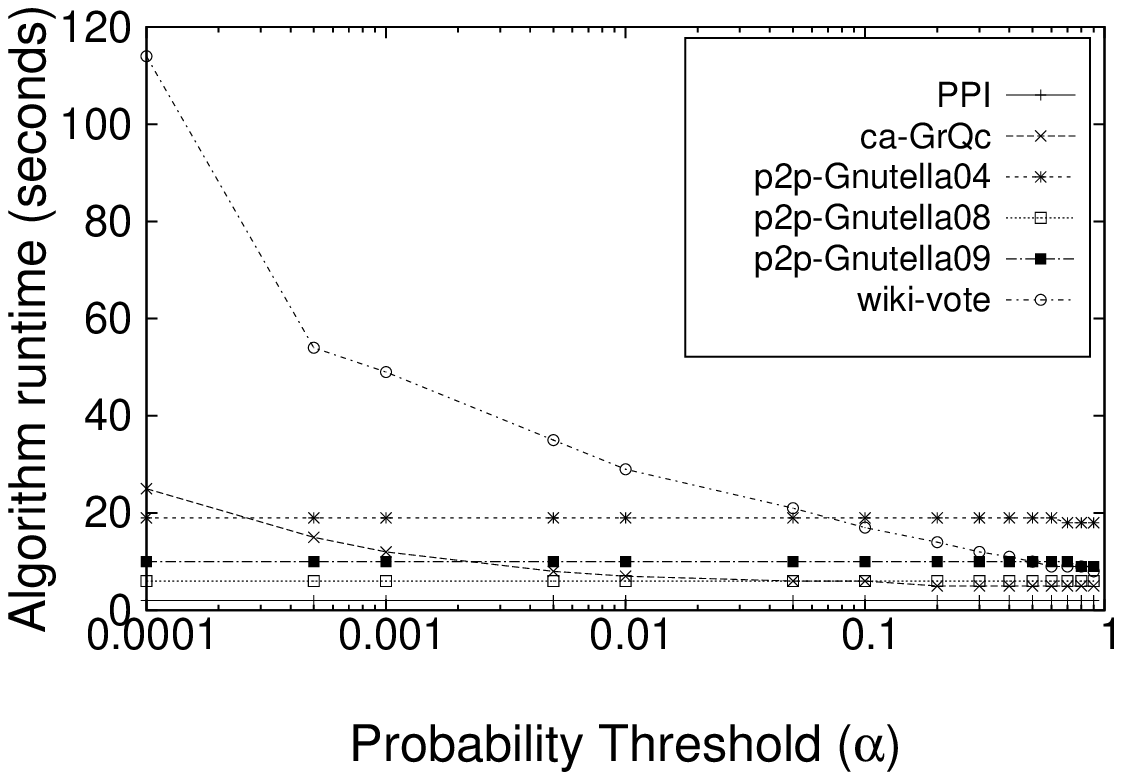}}
\end{array}$
\vspace{-0.15cm}
\caption{Runtime vs Alpha ($\alpha$). The X--Axis is in log--scale} \label{fig:1}
\end{center}
\vspace{-0.5cm}
\end{figure*}

\begin{figure*}[!ht]
\vspace{-0.85cm}
\begin{center}

$\begin{array}{cc}

\subfloat[Random Graphs]{\label{fig:3a}\includegraphics[width=0.35\textwidth]{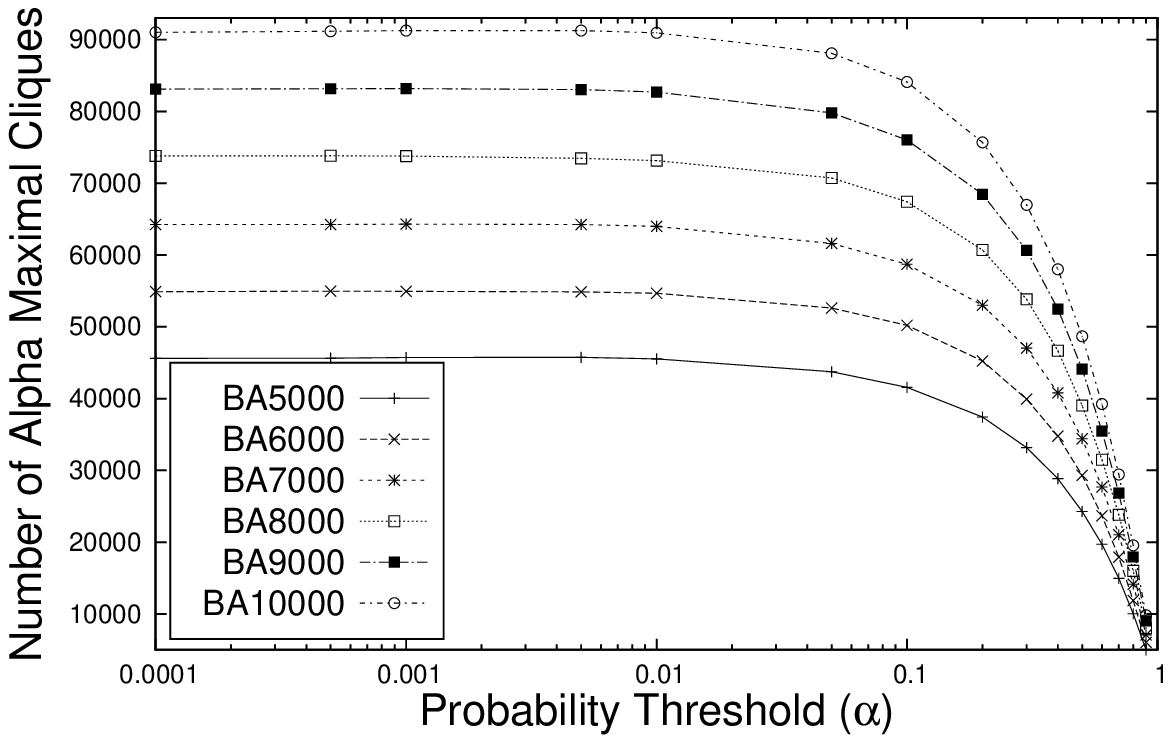}} &

\subfloat[Semi--synthetic and Real Graphs]{\label{fig:3b}\includegraphics[width=0.35\textwidth]{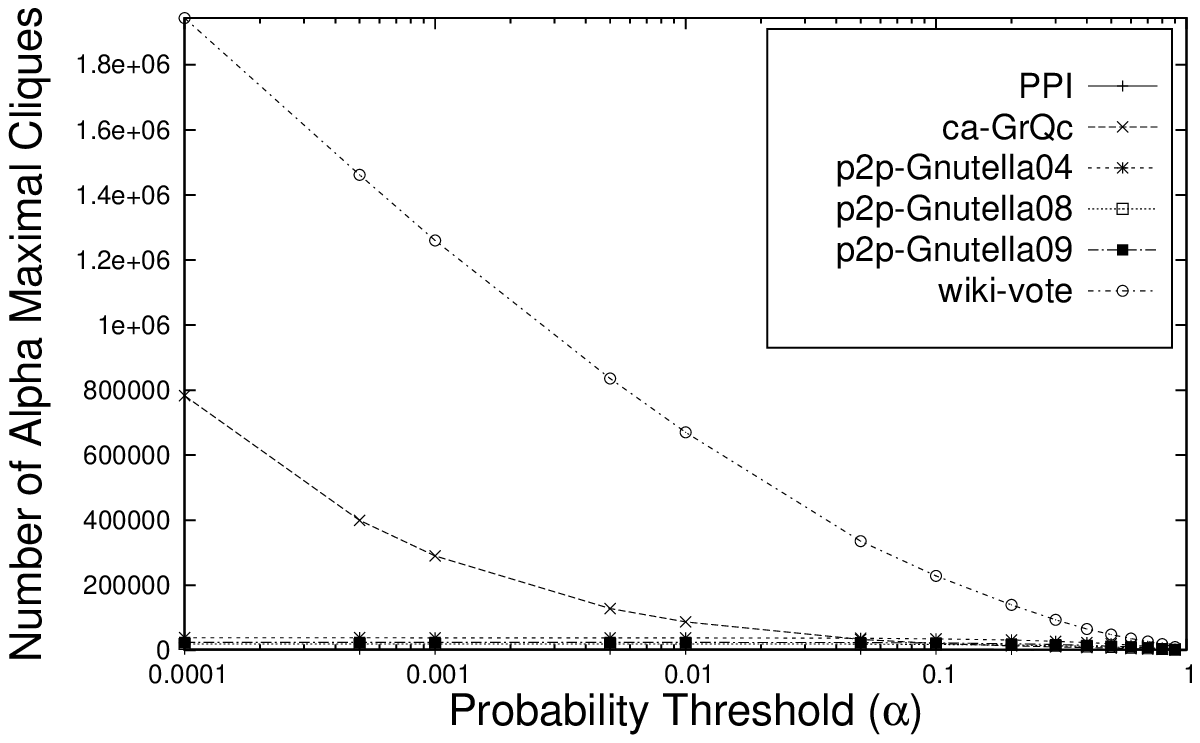}}

\end{array}$
\vspace{-0.25cm}
\caption{No of $\alpha$-maximal cliques vs Alpha ($\alpha$). The X--Axis is in log--scale} \label{fig:3}
\end{center}
\vspace{-0.5cm}
\end{figure*}

\paragraph{Input Data:}
Details of the input graphs that we used are shown in
Table~\ref{table:input}.  

The first set of graphs consists of real world uncertain graphs shared by authors of~\cite{ZLGZ2010}~and~\cite{KBGG2014}. These include a protein-protein interaction (PPI) network of a Fruit Fly obtained by integrating data from the BioGRID~\footnote{http://thebiogrid.org/} database with that form the STRING~\footnote{http://string-db.org/} database, and the DBLP~\footnote{http://dblp.uni-trier.de/} dataset from authors of~\cite{KBGG2014}, which is an uncertain network predicting future co-authorship. The PPI network is an uncertain graph where each vertex represents a protein and two vertices are connected by an edge with a probability representing the likelihood of interaction between the the two proteins. The DBLP network represents co-authorship in academic articles. Each vertex in this network represents an author. Two vertices are connected by an edge with a probability that depends on the ``strength'' of their co-authorship, which is computed as $1- e^{-c/10}$, where $c$ is the number of papers co--authored.

The second set of graphs was obtained from the Stanford Large Network 
Collection~\cite{SNAP}, and includes graphs representing Internet p2p networks, 
collaboration networks, and an online social network.
The p2p-Gnutella graphs represent peer to peer file sharing networks, 
where each vertex in the graph represents a computer and the edges 
represent the communication among them. The p2p-Gnutella04,
p2p-Gnutella08 and p2p-Gnutella09 graphs represent communications
occurring on 4th, 8th and 9th of August, 2002 respectively.
The ca-GrQc graph represents the collaboration network among scientist 
working on General Relativity and Quantum Cosmology. Each vertex in the
graph is a scientist and two vertices are connected by an edge if the 
corresponding scientists have co-authored a paper.
Finally the wiki-vote graph represents the voting that occurs while
selecting a new wikipedia administrator. Each vertex is either a wikipedia
admin or wikipedia user and the edges represent the votes that each 
admin / user casts in favor of a candidate. The candidate is also a
wikipedia user and hence is represented by a vertex in the graph.
For all these graphs, the uncertain graphs were created from these 
deterministic graphs by assigning edge probabilities uniformly at random.  
Hence these can be considered as semi--synthetic uncertain graphs.

The third set of input graphs was
synthetically generated using the Barab{\'a}si$-$Albert model for
random graphs~\cite{RB2002}. Then the edges were assigned
probabilities uniformly at random from $[0,1]$. \\

{\bf Comparison with other approaches.}
We compare our algorithm with another algorithm based on depth-first-search, which we call DFS-NOIP (DFS with NO Incremental Probability Computation), described in Algorithm~\ref{algo:naive}. 
This algorithm also performs a depth first search to enumerate all $\alpha$--maximal cliques but does not compute the probabilities incrementally like MULE does. 


\begin{algorithm}[ht]
\caption{DFS--NOIP($C$,$I$)}
\label{algo:naive}

$I_{copy} \leftarrow I$ \;

\ForAll{$u \in I_{copy}$}
{
    \If{$u \le max(C)$ OR $\cliqueprob(C \cup \{ u \}) < \alpha$}
    {
        $I \leftarrow I \setminus \{ u \}$
    }
}

\If{$I = \emptyset$}
{
    \If{$C$ is an $\alpha$-maximal clique}
    {
        Output $C$ as $\alpha$-maximal clique \;
        \Return \;
    }
}

\ForAll{$v \in I$}
{

    $C' \leftarrow C \cup \{ v \}$ \;

    \If{$C'$ is an $\alpha$-maximal clique}
    {
        Output $C'$ as $\alpha$-maximal clique \;
    }
    \Else
    {
        $I' \leftarrow I \cap \Gamma( v )$ \;
        DFS--NOIP($C'$,$I'$) \;
    }
}
\end{algorithm}

Figure~\ref{fig:0} compares the performance of MULE with DFS--NOIP. The results show that MULE performs much better than DFS--NOIP. For instance, for the graph wiki--vote with $\alpha=0.9$ DFS--NOIP took $64$ seconds while MULE took only $8$ secs. The relative performance results hold true over a wide range of input graphs and values of $\alpha$, including synthetic and real-world graphs, and small and large values of $\alpha$. For $\alpha = 0.0001$, MULE took only $25$ secs to enumerate all maximal cliques in ca-GrQc, while DFS--NOIP took over $4400$ secs. On the wiki--vote input graph with probability threshold 0.9, MULE took $8$ seconds while DFS--NOIP took $64$ seconds. For the same graph, with probability threshold $0.0001$, MULE took 114 secs, while DFS--NOIP took more than 11 hours.

{\bf Dependence on $\alpha$.} We measured the runtime of enumeration
as well as the output size, (the number of $\alpha$-maximal cliques
that were output) for different values of $\alpha$ and for the various
input graphs described above. The dependence of the runtime on
$\alpha$ is shown in Figure~\ref{fig:1}, and the number of cliques as
a function of $\alpha$ is shown in Figure~\ref{fig:3}. We note that as
$\alpha$ increases, the number of maximal cliques, and the time of
enumeration both drop sharply. The decrease in runtime is because with
a larger value of $\alpha$, the algorithm is able to prune search
paths aggressively early in the enumeration.

We note that the number of $\alpha$-maximal cliques does not have to
always decrease as $\alpha$ increases. Sometimes it is possible that
the number of $\alpha$-maximal cliques increases with $\alpha$. This
is because as $\alpha$ increases, a large maximal clique may split
into many smaller maximal cliques. However, these differences are
negligible, and are not visible in the plots.

{\bf Dependence on Size of Output.}  Figure~\ref{fig:4} shows the change in runtime with respect to the number of $\alpha$-maximal cliques enumerated, for the randomly generated graphs. It can be seen that the runtime of the algorithm is almost proportional to the number of maximal cliques in the output. This shows that the algorithm runtime scales well with the number of $\alpha$-maximal cliques in output. This comparison was not done for real world or semi--synthetic graphs as these graphs have different structural properties, hence different sizes of maximal cliques and thus there is no meaningful way to interpret the results.

\begin{figure}[!ht]
\vspace{-0.5cm}
\begin{center}


\subfloat[Random Graphs]{\label{fig:4a}\includegraphics[width=0.35\textwidth]{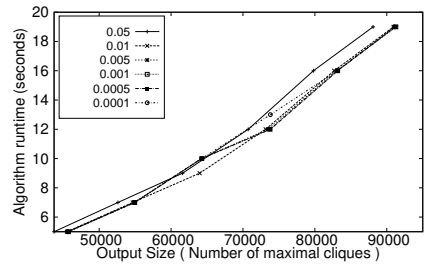}} 



\vspace{-0.2cm}

\caption{Runtime vs Output Size} \label{fig:4}
\end{center}
\vspace{-0.5cm}
\end{figure}

{\bf Enumerating Large Maximal Cliques.} Figures~\ref{fig:5}~and~\ref{fig:6} show the runtime of LARGE--MULE (Algorithm~\ref{algo:largemule}) and the output size respectively as a function of $t$, the minimum size of an $\alpha$-maximal clique that is output. As $t$ increases, both runtime and output size decrease substantially. For instance, MULE takes $76797$ seconds to enumerate all uncertain maximal cliques from the DBLP dataset (for probability threshold $0.9$).  However, LARGE--MULE takes only $32$ seconds when $t=3$. Similarly, for input graph ca-GrQc and $\alpha=0.0001$, MULE takes $125$ seconds, while LARGE--MULE takes 10 seconds when $t=6$ and 6 seconds when $t=7$.

\begin{figure*}[!ht]
\vspace{-1cm}
\begin{center}

$\begin{array}{cc}

\subfloat[BA10000]{\label{fig:5a}\includegraphics[width=0.35\textwidth]{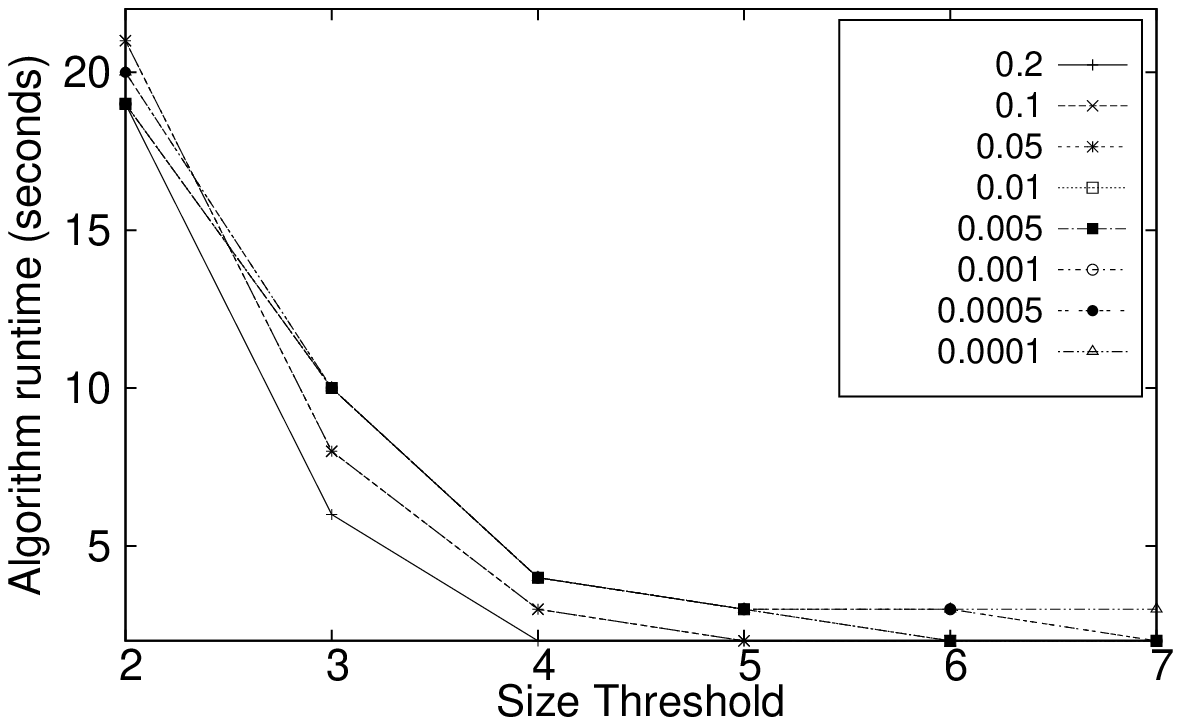}} &

\subfloat[ca-GrQc]{\label{fig:5b}\includegraphics[width=0.35\textwidth]{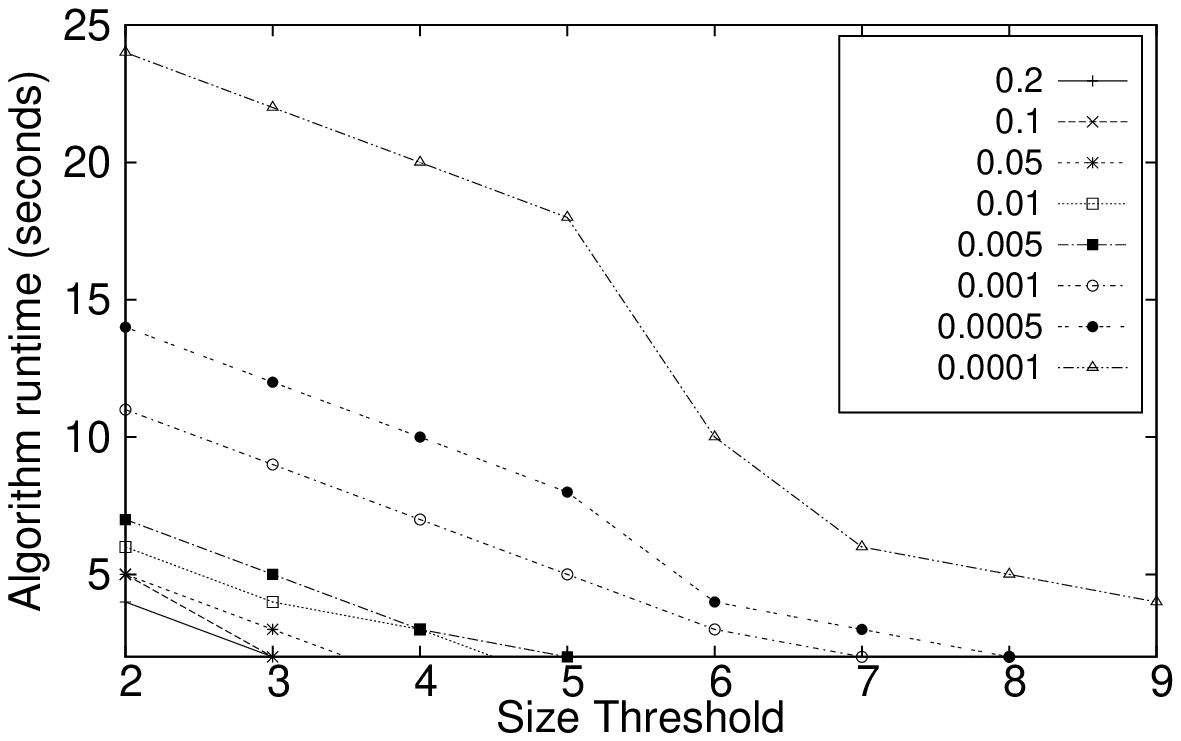}}



\subfloat[DBLP]{\label{fig:5c}\includegraphics[width=0.35\textwidth]{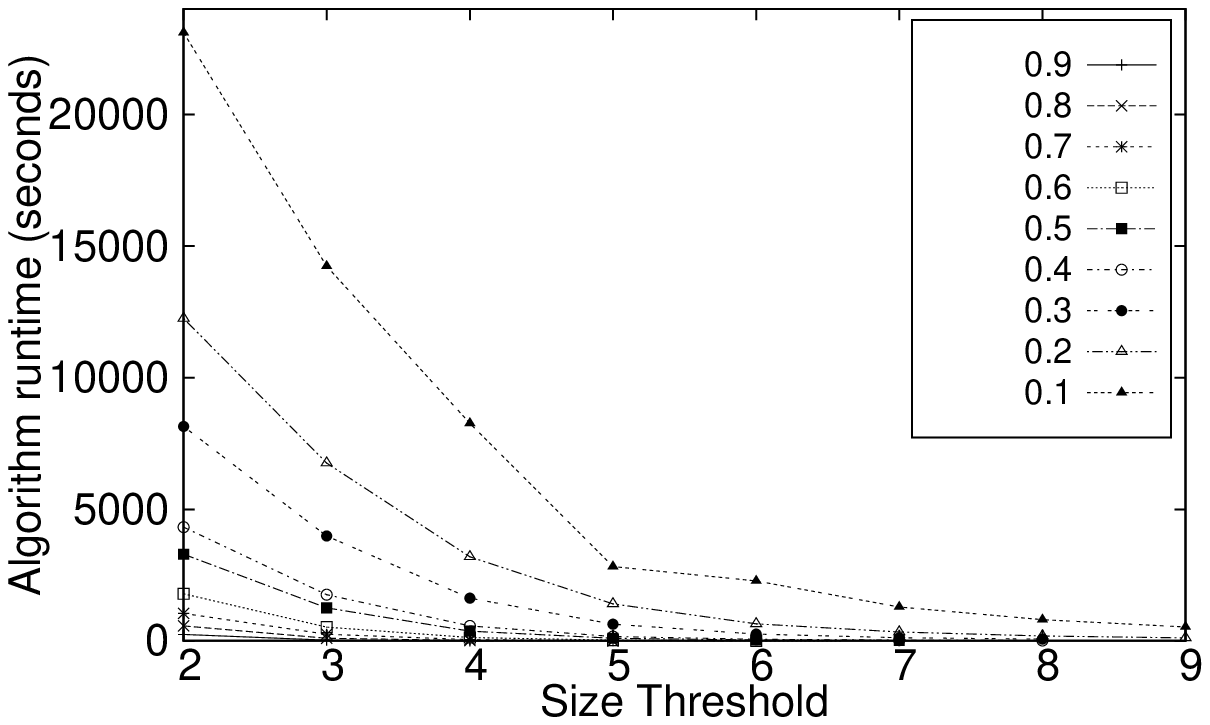}} 
\end{array}$

\vspace{-0.2cm}

\caption{Runtime vs Size threshold of enumerated uncertain maximal cliques} \label{fig:5}
\end{center}
\vspace{-0.5cm}
\end{figure*}

\begin{figure*}[!ht]
\vspace{-1cm}
\begin{center}

$\begin{array}{ccc}

\subfloat[BA10000]{\label{fig:6a}\includegraphics[width=0.35\textwidth]{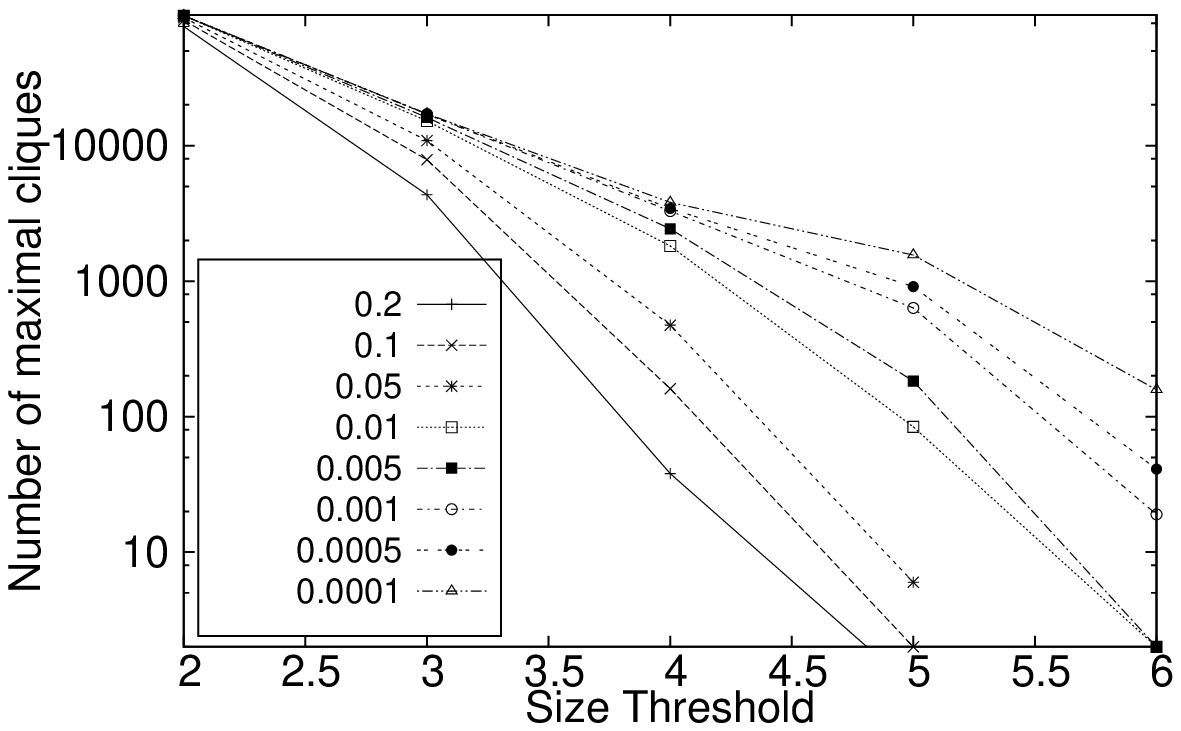}} &

\subfloat[ca-GrQc]{\label{fig:6b}\includegraphics[width=0.35\textwidth]{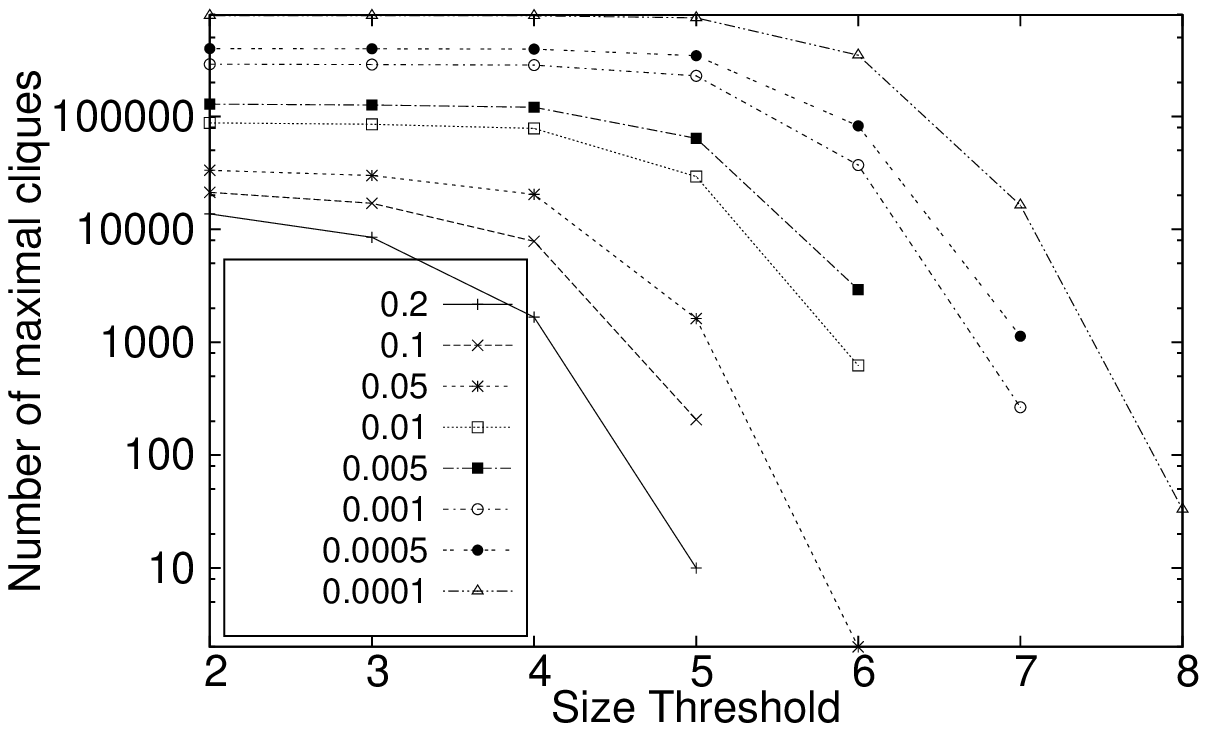}}


\subfloat[DBLP]{\label{fig:6c}\includegraphics[width=0.35\textwidth]{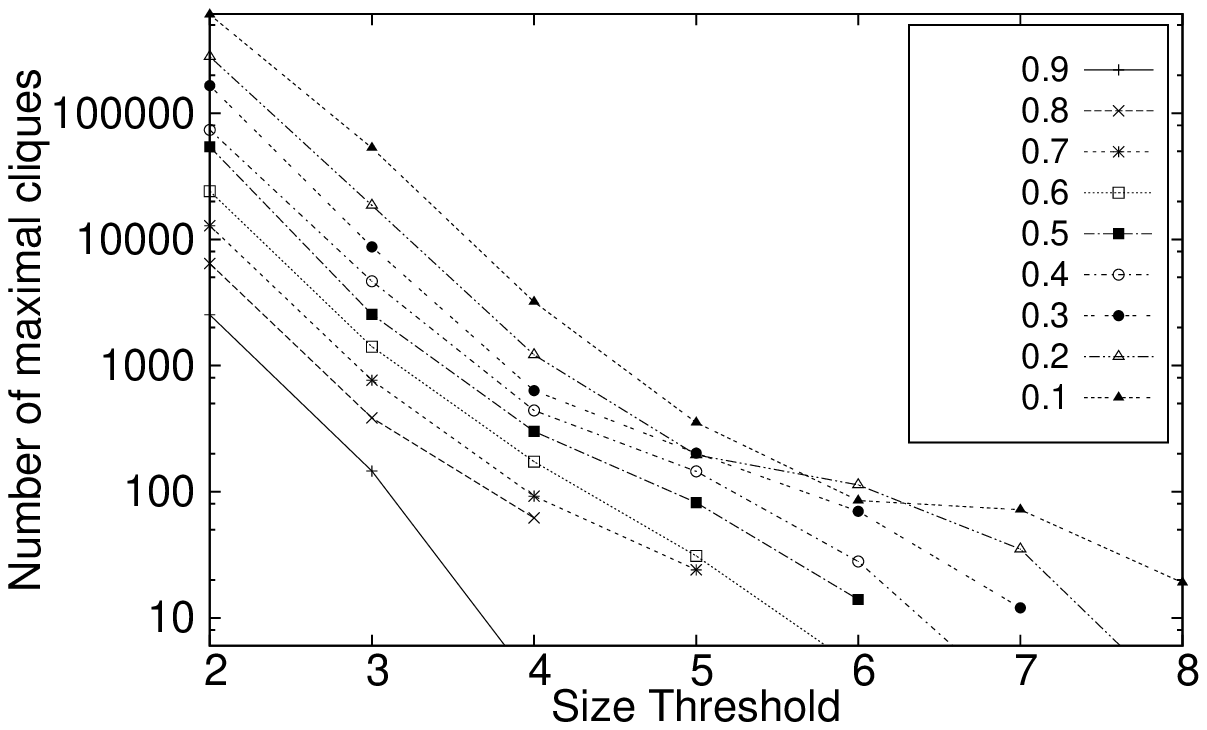}} 
\end{array}$

\vspace{-0.2cm}

\caption{Number of $\alpha$-maximal cliques vs threshold on minimum size of uncertain maximal clique} 
\label{fig:6}
\end{center}
\vspace{-0.5cm}
\end{figure*}


\remove{
\item From Figure~\ref{fig:2}, we can see that the algorithm runtime
  increases almost linearly with input size for random graphs.
  This comparison was not performed for semi--synthetic or real
  world graphs as they have different underlying structures and
  hence the comparison does not give any meaningful interpretation.
}

%

\section{Conclusion}

We present a systematic study of the enumeration of maximal cliques from an uncertain graph, starting from a precise definition of the notion of an $\alpha$-maximal clique, followed by a proof showing that the maximum number of $\alpha$-maximal cliques in a graph on $n$ vertices is exactly $n \choose \lfloor n/2 \rfloor$, for $0 < \alpha < 1$. We present a novel algorithm, MULE, for enumerating the set of all $\alpha$-maximal cliques from a graph, and an analysis showing that the worst-case runtime of this algorithm is $O \left ( n \cdot 2^n \right)$. We present an experimental evaluation of MULE showing its performance, and an extension for faster enumeration of large maximal cliques.

An interesting open problem is to design an algorithm for enumerating maximal cliques from an uncertain graph whose time complexity is worst-case optimal, $O \left ( \sqrt{n} \cdot 2^n \right)$.
Finally, there are various dense substructures that can be found in a network. Some examples include
bicliques, quasi--cliques and k-cores. Finding these dense substructures in the context of uncertain graphs
can be an important future direction of work.

\bibliographystyle{plain}
\bibliography{references}

\end{document}